\documentclass[aps,prx,10pt,superscriptaddress,twocolumn,nofootinbib,floatfix]{revtex4-2}
\usepackage{graphicx}
\usepackage{amsmath,amsthm,amssymb,dsfont,soul,enumitem}
\usepackage{mdframed}
\usepackage{orcidlink}
\usepackage{array}
\usepackage{tikz}
\usetikzlibrary{arrows.meta,positioning,calc,decorations.pathreplacing}
\newcommand{\ket}[1]{\left| #1 \right\rangle}
\newcommand{\bra}[1]{\left\langle #1 \right|}
\newcommand{\proj}[1]{\ket{#1}\hskip-1mm\bra{#1}}

\newcommand{\braket}[2]{\langle #1|#2 \rangle}
\newcommand{\ketbra}[2]{\left|#1\right\rangle\hskip-1mm\left\langle#2\right|}

\usepackage{hyperref}
\usepackage[capitalise]{cleveref}
\crefname{equation}{Eq.}{Eqs.}
\crefname{figure}{Fig.}{Figs.}
\soulregister{\cite}{7}
\soulregister{\cref}{7}
\soulregister{\Cref}{7}
\soulregister{\ref}{7}
\newtheorem{theorem}{Theorem}
\newtheorem{corollary}{Corollary}
\newtheorem{lemma}{Lemma}
\newtheorem{proposition}{Proposition}
\newcommand{\Fs}{F_{\mathrm{s}}}
\newcommand{\Nneg}{\mathcal{N}}
\newcommand{\KD}{\varrho_{\mathbb C}}
\newcommand{\wv}[1]{\langle\hat\Pi(a)\rangle_{w,#1}}
\newcommand{\panel}[2]{\begin{tikzpicture}\node[anchor=south west,inner sep=0] (img) {\includegraphics[width=0.44\textwidth]{#2}};\node[anchor=north east,font=\small] at ([xshift=-2pt,yshift=-2pt]img.north west) {(#1)};\end{tikzpicture}}

\makeatletter \def\switch@array{}\makeatother

\begin{document}

\title{Negative quasiprobabilities redistribute (and complex ones reduce) the information in unabsorbed photons}

\author{Jonte R. Hance\,\orcidlink{0000-0001-8587-7618}}
\email{jonte.hance@newcastle.ac.uk}
\affiliation{Quantum Group, School of Computing, Newcastle University, 1 Science Square, Newcastle upon Tyne, NE4 5TG, UK}

\begin{abstract}
Negative Kirkwood--Dirac quasiprobabilities are credited with the quantum advantage of interaction-free measurement and postselected metrology. We test that claim on the task interaction-free measurement was designed for: learning about a sample while it absorbs as little light as possible. The quasiprobability distribution of path and output gives the linear response to a weak sample: its real part responds to absorption, its imaginary part to a phase. For a pure single photon passing once through any interferometer, the Fisher information about transmittance per absorbed photon is the direct-transmission value minus a sum of squared imaginary parts: no interferometer or readout beats direct transmission. Negativity only concentrates information in weakly lit outputs. These imaginary parts depend on the output basis, not the sample, so an optimal basis with all terms real always exists. A single photon beats direct transmission only through a phase response, which for an optical line needs a thick or inhomogeneously broadened sample, or through many passes when absorptions go uncounted, as a critically coupled resonator shows. For detecting a weak absorber, negativity gives a single-shot advantage that does not accumulate over repeated trials. The advantage that accumulates comes from outputs dark without the object, and needs no negativity.
\end{abstract}

\maketitle

\section{Introduction}

Measuring a sample that the probe damages is a problem of rationing: we want to maximise the information gained per particle absorbed, not per particle sent. This limits the imaging of biological matter with electrons and X-rays \cite{Henderson1995}, and is the problem interaction-free measurement promises to resolve, since an interferometer can detect an object with photons that never reach it \cite{Elitzur1993,Kwiat1995,Kwiat1999}. Interaction-free measurement has been extended to imaging \cite{White1998,Zhang2019,Hance2021} and proposed for electron microscopy \cite{PutnamYanik2009,Kruit2016}, where it has been demonstrated with electrons \cite{Turner2021}. Quantum foundations has meanwhile proposed a candidate for what provides interaction-free measurement's quantum ingredient: negative values of the Kirkwood--Dirac (KD) quasiprobability \cite{Kirkwood1933,Dirac1945,ArvidssonShukur2024}. Negativity of the doubly extended distribution (the KD distribution with the generator's eigenbasis entered twice) is necessary for postselected metrology to exceed the classical bound per retained particle \cite{ArvidssonShukur2020}, the resulting advantage has been demonstrated experimentally \cite{LupuGladstein2022}, and the enhancement has been traced to an anomalous conditional uncertainty of the generator \cite{Ji2026}. KD negativity also allows lossless distillation of the Fisher information about several parameters \cite{Jenne2022}, and sufficiently negative values violate noncontextual inequalities in a three-path interferometer \cite{Hofmann2023,Hofmann2025}, where contextuality requires counterfactual gain \cite{Sagawa2026}. For interaction-free measurement, a negative KD term is sufficient for counterfactual gain, the quantum part of the change an absorber makes to the output statistics \cite{Hance2024}. That paper computed it for an opaque object and a single photon's guess. Whether that negativity helps the task interaction-free measurement is meant for, gentle (high information-per-absorbed-photon) measurement of a partially transparent object, has not been settled.

Answering this needs the limits on the information per absorbed photon, and these depend on what is recorded. If the source has a known mean and an absorption is not itself detected, only the output counts hold information, and we call the absorptions \emph{unregistered}. If the photon number is known, an absorption is itself an outcome, and we call them \emph{registered} (the literature also says that the absorptions are counted, or the photon heralded). The bounds on information per absorbed photon differ between the two cases by the factor $1/(1-t)$, $t$ being the object's intensity transmittance, and which case applies decides whether a multi-pass scheme can help. In both, our baseline for comparison is direct transmission: sending the light straight through the object and counting what arrives.

Without an interferometer, the limits are known: a coherent state provides information $N/t$ \cite{Monras2007} and an $N$-photon number state $N/[t(1-t)]$ \cite{Adesso2009}, which no ancilla-assisted parallel probe exceeds (sequential strategies being left open \cite{Nair2018}), and which a heralded single photon attains in principle and approaches in experiment \cite{Whittaker2017}. For interaction-free schemes, Mitchison and Massar showed that two partially transparent objects cannot be distinguished absorption-free with certainty~\cite{Mitchison2001}, and with Pironio they bounded the number of photons any protocol must sacrifice to resolve a difference in transparency, a bound~\cite{Massar2001,Mitchison2002} that counting the absorbed photons matches up to a constant for close transparencies, unless the two transparencies differ only in phase~\cite{Massar2001}. Facchi \textit{et al} found the same Cram\'er--Rao bound per absorbed particle for standard and Zeno tomography with binomial statistics, and an advantage for Zeno tomography in telling apart grey levels that are unequally represented in the sample~\cite{Facchi2002}. Wildfeuer's recent bound covers any fixed arrangement of one photon and a memoryless, non-dispersive object, however many passes are used: with the absorptions registered, no arrangement beats direct transmission per absorbed photon, and the Elitzur--Vaidman interferometer only attains this bound, so its advantage for detecting an object does not extend to estimation~\cite{Wildfeuer2026}.

Other work reports gains under different normalisations. Thomas \textit{et al} found by simulation that a Zeno interferometer measuring a transparency loses no fewer probe particles than a transmission measurement when the losses are counted, but fewer with a Poissonian source, unregistered losses and a transparency above about one half, and fewer, counted or not, for telling apart two transparencies of high contrast~\cite{Thomas2014}. Twin-beam probes gain information per photon incident on the sample, and so, in a single pass, per photon absorbed \cite{Brida2010,Moreau2017,Losero2018}, and a coherent state in a ring resonator gains information per photon entering the device \cite{Belsley2022} (\cref{sec:discussion} compares both with the bounds we give later).

These works identify readouts that reach their own limits, intensity detection of twin beams \cite{Losero2018} and of a critically coupled ring's output \cite{Belsley2022}, but not which readouts of an interferometer's unabsorbed photons reach the limits per absorbed photon, which is where the KD terms would enter. The conditions under which a readout of a pure-state family attains the quantum Fisher information are standard \cite{Braunstein1994,BarndorffNielsen2000}: for a phase, a readout that falls short loses four times the variance of the real weak values of the generator \cite{Hofmann2011}. None of the work above evaluates that loss for an absorber or connects it to the quantities the interaction-free literature works with. Wildfeuer's proof evaluates the quantum Fisher information of the purified state, which bounds the photon's information by monotonicity under the partial trace, and so covers every readout at once. It does not, however, tell us which readouts of the unabsorbed photons reach the bound in general, what a readout that falls short loses, how either relates to the KD terms the interaction-free literature works with, or (beyond the example of classical light passed many times through the object) what changes when the absorptions go unregistered. Those are the questions we address here.

We settle whether negativity raises the information per absorbed photon, for single photons and for coherent light, for a partially transmitting object and with every figure of merit normalised per absorbed photon. The KD distribution of path and output turns out to be the outputs' linear response to a weak sample (\cref{sec:setting}, \cref{fig:kdplane}): its real part is each output's response to absorption in the path, its imaginary part is the response to a phase. For one photon in a pure state passing once through any interferometer, \cref{thm:main} and \cref{cor:absorbed} then give the Fisher information about an absorber's transmittance per absorbed photon as the direct-transmission value minus a sum of squared imaginary parts. No single-pass interferometer or readout therefore beats direct transmission per absorbed photon, in either case: the registered bound is Wildfeuer's, and the unregistered bound follows from this bound for a single pass, since the absorbed outcome adds $P(a)/(1-t)$ to the Fisher information whatever the readout ($P(a)$ being the probability that, if measured, the photon would be on the object's path). Negativity plays no role in the total: it only moves information between outputs, concentrating it in weakly lit ones, which is weak-value amplification of an absorber (\cref{cor:detected}). The imaginary parts are set by the readout rather than the object, so for a pure input a readout with real KD terms always exists, and we show every such readout is optimal; the bound also holds for mixed inputs and for several absorbers at once, though a mixed input no longer reaches this bound with every real readout (\cref{sec:extensions}). The identity also shows the two routes by which single-photon interferometry can gain information per absorbed photon (\cref{sec:helps}). A phase response turns the KD's imaginary parts from loss into signal, which gains information over direct transmission for an optically thick or inhomogeneously broadened sample (\cref{thm:dispersive}, \cref{sec:dispersive}). Interference between the responses of different passes gains information when the absorptions go unregistered, and a critically coupled ring resonator on resonance uses this to come arbitrarily close to the registered bound for a long wavepacket (\cref{prop:manypass}, \cref{sec:multipass}). For deciding whether a weak absorber is there at all, negativity gives a single-shot advantage per absorbed photon that does not accumulate over repeated or sequential trials. The advantage that accumulates over trials comes from outputs that are dark, or nearly dark, when the object is absent, and the optimal single-pass test has no negative KD term. Only for a strongly absorbing object does negativity raise the exponent of a repeated test (\cref{sec:discrimination}).

\Cref{sec:setting} sets out the linear response and the weak-value form of the statistics. \Cref{sec:fisher} proves the Fisher information identity, and extends it to mixed inputs and several objects. \Cref{sec:helps} discusses the two routes which provide a gain: a dispersive object, and many passes. \Cref{sec:discrimination} turns to deciding whether the object is present, \cref{sec:test} proposes a test of \cref{cor:detected} on an existing silicon three-path interferometer and of \cref{cor:absorbed} on a higher-visibility version of it, and \cref{sec:discussion} returns to the question of whether we can call negativity or imaginarity a resource for imaging, and what the results mean for low-dose imaging. \Cref{tab:map} summarises the bounds we and others have given.

\section{Setting, linear response, and the weak-value form of the statistics}
\label{sec:setting}

\begin{figure}[t]
\centering
\begin{tikzpicture}[scale=0.9, every node/.style={font=\small}]
  % interferometer body
  \draw[thick, fill=black!5] (2,0) rectangle (5,3);
  \node at (3.5,1.5) {\parbox{2.4cm}{\centering lossless optics\\(any $d$-path\\interferometer)}};
  % input
  \draw[thick,-{Latex}] (0.3,1.5) -- (2,1.5);
  \node[above] at (1.1,1.55) {$\ket{\psi}$};
  % paths inside (schematic): the absorber path a
  \draw[thick] (2.5,0) -- (2.5,-0.9) -- (4.5,-0.9) -- (4.5,0);
  \draw[thick, fill=white] (3.2,-1.2) rectangle (3.8,-0.6);
  \node[below] at (3.5,-1.25) {path $\ket{a}$: transmittance $t$};
  % outputs
  \foreach \y/\lab in {2.6/{$\ket{m_1}$}, 1.9/{$\ket{m_2}$}, 0.5/{$\ket{m_d}$}} {
    \draw[thick,-{Latex}] (5,\y) -- (6.6,\y);
    \node[right] at (6.6,\y) {\lab};
  }
  \node at (5.8,1.25) {$\vdots$};
  % classification of outputs
  \node[draw, rounded corners, inner sep=3pt, font=\footnotesize, align=left] at (3.7,-2.5) {for $t<1$, $m$ brightens iff $\varrho(a,m)<\tfrac12(1-\sqrt t)|\braket{m}{a}|^2P(a)$\\ as $t\to1$: $\varrho(a,m)>0$: darkens \quad $\varrho(a,m)<0$: favours presence\\ $P(m)=0$, $|\braket{m}{a}|^2>0$: null port, lit for $t<1$: proves presence};
\end{tikzpicture}
\caption{A single photon in state $\ket{\psi}$ enters an arbitrary lossless $d$-path interferometer. One internal path, $\ket{a}$, contains an object of intensity transmittance $t$; the photon is detected in an output basis $\{\ket{m}\}$ that includes all optics after the object. Each output is classified by its Kirkwood--Dirac term $\varrho(a,m)=\mathrm{Re}[\braket{m}{a}\bra{a}\hat\rho\ket{m}]$ against its Elitzur--Vaidman term $|\braket{m}{a}|^2P(a)$. For $t<1$ an output brightens when the object is inserted if and only if $\varrho(a,m)<\tfrac12(1-\sqrt t)|\braket{m}{a}|^2P(a)$, which is \cref{eq:brightcond}, with $P(a)=|\braket{a}{\psi}|^2$ the probability of path $a$ with the object absent; no model that only removes photons from path $a$ allows any output to brighten. The two signs shown are that condition as $t\to1$: away from that limit an output with a small positive $\varrho(a,m)$ brightens too, and a null port, dark in the object's absence but reachable from $\ket{a}$, brightens at every $t<1$, so a click there proves the object present, though not, by itself, contextuality (\cref{sec:discussion}).}
\label{fig:setup}
\end{figure}
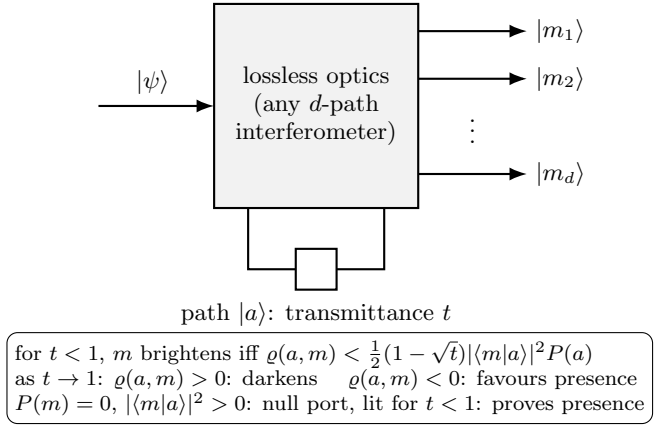

A single photon is prepared in a state $\ket{\psi}$ on $d$ paths and passes once through lossless linear optics before being detected in an orthonormal output basis $\{\ket{m}\}$; any optics after the object are included in the choice of that basis (\cref{fig:setup}). Single-pass means throughout the interferometer that the photon meets the object at most once, so that the object enters the evolution only through the operator $K_t$ below. One internal path $\ket{a}$ contains a non-dispersive object of intensity transmittance $t\in[0,1]$, so the amplitude on that path is multiplied by $\sqrt{t}$ (a constant phase on that path is equivalent to a unitary acting after the object, so belongs to the optics included in the choice of basis). The unabsorbed part of the state is $K_t\ket{\psi}$ with
\begin{equation}
\begin{split}
K_t &= \mathds{1}-(1-\sqrt{t})\proj{a},\\
S(t)&=\bra{\psi}K_t^\dagger K_t\ket{\psi}=1-(1-t)P(a),
\end{split}
\end{equation}
where $P(a)=|\braket{a}{\psi}|^2$, and $S(t)$ is the survival probability (i.e., the object absorbs $(1-t)P(a)$ photons per trial). The output amplitudes and probabilities are
\begin{equation}
z_m(t) = \braket{m}{\psi}-(1-\sqrt{t})\braket{m}{a}\braket{a}{\psi}, \label{eq:amp}
\end{equation}
\begin{equation}
\begin{split}
P(m|t) = |z_m(t)|^2 = &P(m)-2(1-\sqrt{t})\varrho(a,m)\\ &+(1-\sqrt{t})^2|\braket{m}{a}|^2P(a), \label{eq:Pmt}
\end{split}
\end{equation}
where $P(m)=|\braket{m}{\psi}|^2$ is the probability with the object absent and, in the notation of Ref.~\cite{Hance2024} extended here to the complex term,
\begin{equation}
\begin{split}
&\varrho(a,m)=\mathrm{Re}[\KD(a,m)],\\
&\KD(a,m) = \braket{m}{a}\hspace{-3pt}\bra{a}\hat\rho\ket{m},
\label{eq:kd}
\end{split}
\end{equation}
with $\hat\rho=\proj\psi$, where $\varrho(a,m)$ is the real part of the KD quasiprobability $\KD(a,m)$ for path $a$ and output $m$ \cite{Kirkwood1933,Dirac1945}. $\varrho(a,m)$ is the closest thing to an answer that quantum mechanics gives to the question ``how much of output $m$ came through path $a$'', and can be negative. At $t=0$, \cref{eq:Pmt} is Ref.~\cite{Hance2024}'s counterfactual output distribution $P(m|X_a)$, Eq.~(5) there, $X_a$ denoting a blocked path $a$, and the last term of \cref{eq:Pmt} is that paper's Elitzur--Vaidman term, scaled by $(1-\sqrt{t})^2$. Two sum rules follow from completeness of the output basis and (for the second) purity of the input:
\begin{equation}
\begin{split}
\sum_m \KD(a,m)&=P(a),\\
\sum_{m:\,P(m)>0}\frac{|\KD(a,m)|^2}{P(m)} &= P(a)\bra{a}\hat\Pi_+\ket{a},
\label{eq:sumrules}
\end{split}
\end{equation}
where $\hat\Pi_+=\sum_{m:\,P(m)>0}\proj{m}$ projects onto the outputs which have non-zero probability of being occupied when the object is absent. The first sum rule says the real parts $\varrho(a,m)$ sum to $P(a)$ and the imaginary parts to zero; the second follows from $|\KD(a,m)|^2=P(m)|\braket{m}{a}|^2P(a)$, and its right-hand side is $P(a)$ when no unoccupied output is reachable from $\ket{a}$, in which case the rule is Hosoya and Shikano's weak-value variance identity \cite{Hosoya2010} for the projector $\hat\Pi(a)$. Outputs with $P(m)=0$ but $\braket{m}{a}\neq0$ are the dark or null ports, which receive photons only when the object is present: each takes $|\braket{m}{a}|^2P(a)$ out of the second sum rule. Dividing the KD term by $P(m)$ gives the weak value of the path projector $\hat\Pi(a)=\proj{a}$ between preselection $\ket\psi$ and postselection $\ket m$ \cite{Aharonov1988}, in Ref.~\cite{Hance2023}'s notation, with the postselected output $m$ added as a second subscript:
\begin{equation}
\wv{m} = \frac{\bra{m}\hat\Pi(a)\ket{\psi}}{\braket{m}{\psi}} = \frac{\KD(a,m)}{P(m)} .
\label{eq:wv}
\end{equation}
Classically this would be a conditional probability, so $0\le\mathrm{Re}[\wv{m}]\le 1$. A real part below $0$ or above $1$ is therefore an anomalous weak value, which requires both $\ket\psi$ and $\ket m$ to be coherent between $\ket{a}$ and its complement \cite{Wagner2023} (and is a proof of contextuality \cite{Pusey2014}, tied in a three-path interferometer to the violation of a noncontextual inequality \cite{Hofmann2023}). Dividing \cref{eq:amp} by $\braket{m}{\psi}$ shows that the whole effect of the object on a bright output is set by this weak value:
\begin{equation}
\frac{P(m|t)}{P(m)}=\big|1-(1-\sqrt{t})\,\wv{m}\big|^2 .
\label{eq:LR}
\end{equation}
The left-hand side is the likelihood ratio of a click in output $m$ between the object present and absent, so every statistical quantity in this paper is a functional of the distribution $\{P(m),\wv{m}\}$. This is constrained by $\sum_{m:\,P(m)>0}P(m)\wv{m}=P(a)$ and, from the second sum rule, by $\sum_{m:\,P(m)>0}P(m)|\wv{m}|^2=P(a)\bra{a}\hat\Pi_+\ket{a}\le P(a)$. At $t=0$ the right-hand side of \cref{eq:LR} is $|1-\wv{m}|^2$, the squared modulus of the weak value $1-\wv{m}$ of the projector on the unblocked paths.\footnote{This is the relation, Eq.~(13) of Ref.~\cite{Masiello2026}, by which Masiello \textit{et al} obtain the squared moduli of path weak values from blocked-path intensities in a neutron interferometer, combining them with phase-shifted intensities to fix the real parts; Ref.~\cite{Ijaz2025} noted a link of the same kind between interrogation probabilities and weak values of the non-unitary absorber operator for a polarisation interferometer.}

Differentiating \cref{eq:Pmt} with respect to $t$ gives, for all $t\in(0,1]$,
\begin{equation}
\frac{\partial P(m|t)}{\partial t}=\frac{\varrho(a,m)-(1-\sqrt{t})|\braket{m}{a}|^2P(a)}{\sqrt{t}},
\label{eq:linres}
\end{equation}
so that $\partial_t P(m|t)|_{t=1}=\varrho(a,m)$. The real KD term is therefore the linear response of output $m$ to a weak attenuation of path $a$, and the real weak value $\mathrm{Re}[\wv{m}]=\partial_t\ln P(m|t)|_{t=1}$ is the fractional response. This is the absorption analogue of the statement that the imaginary weak value of a generator is, up to a factor of two, the logarithmic derivative of an output probability with respect to a phase \cite{Hofmann2011,Dressel2012}. To first order in $1-\sqrt{t}$, and for a single postselected output, \cref{eq:LR} is the relation through which Denkmayr \textit{et al} read a path weak value from the intensity change caused by a weak absorber in a matter-wave interferometer \cite[Eq.~(6)]{Denkmayr2014}; we add the complete output basis, with the sum rules \cref{eq:sumrules}, and the imaginary part, whose contribution they neglected. Lostaglio showed that a unitary linear response read this way, as the imaginary part of a KD quasiprobability, is contextual whenever the perturbation is reversible by mixing \cite{Lostaglio2020}. The absorption version also makes the KD distribution operational for a sensor: the sensor measures $\varrho(a,m)$, one output at a time, and negative entries are outputs that brighten when the path is attenuated. Summing \cref{eq:linres} over $m$ necessarily recovers $\partial_t S=P(a)$.

The imaginary part of the same term is the response to a phase. Were the object also to shift the phase of path $a$ by an amount that varies with what is measured, that is, to be dispersive (\cref{sec:dispersive}), it would multiply the amplitude there by $g=\sqrt{t} e^{i\phi}$. \Cref{eq:amp} then would hold with $1-\sqrt{t}$ replaced by $1-g$, so that $P(m|g)=P(m)-2\,\mathrm{Re}[(1-g)\KD(a,m)]+|1-g|^2|\braket{m}{a}|^2P(a)$, and to first order in $1-\sqrt{t}$ and $\phi$, where $1-g\simeq(1-\sqrt{t})-i\phi$,
\begin{equation}
P(m|g)\simeq P(m)-2(1-\sqrt{t})\varrho(a,m)-2\phi\mathrm{Im}[\KD(a,m)].
\label{eq:response}
\end{equation}
The KD term is therefore the whole linear response of output $m$ to a weak object in path $a$ (\cref{fig:kdplane}): its real part is the response to absorption, \cref{eq:linres}, and its imaginary part the response to a phase, Ref.~\cite{Hofmann2011}'s statement for the generator $\hat\Pi(a)$. By the first sum rule, all outputs' responses to absorption add to $P(a)$ (the light the object removes), and their responses to a phase add to zero, since a phase removes no light. No attenuator that only removes photons (rather than also (re)moving coherence between paths) can produce a negative real part; an imaginary part is a response to something an attenuator does not do. For an object that only absorbs, \cref{thm:main} below shows that negative real parts leave the information unchanged and imaginary parts reduce it.

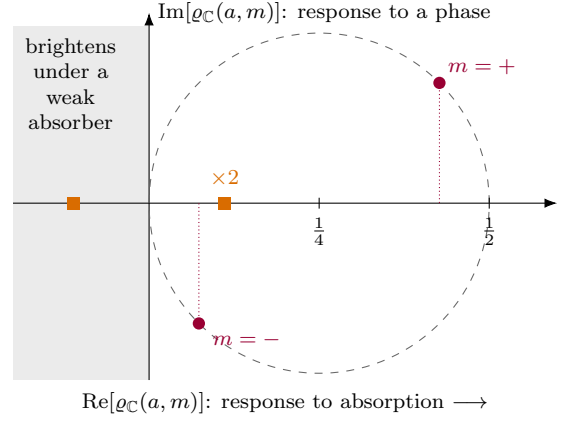
\begin{figure}[t]
\centering
\begin{tikzpicture}[x=9cm,y=9cm,every node/.style={font=\footnotesize}]
  \fill[black!8] (-0.2,-0.26) rectangle (0,0.26);
  \node[align=center,anchor=north west] at (-0.195,0.255) {brightens\\under a\\weak\\absorber};
  \draw[-{Latex}] (-0.2,0) -- (0.6,0);
  \node[anchor=north] at (0.2,-0.265) {$\mathrm{Re}[\KD(a,m)]$: response to absorption $\longrightarrow$};
  \draw[-{Latex}] (0,-0.26) -- (0,0.28) node[right,align=left] {$\mathrm{Im}[\KD(a,m)]$: response to a phase};
  \draw[dashed,black!60] (0.25,0) circle (0.25);
  \foreach \x in {0.25,0.5} \draw (\x,0.008) -- (\x,-0.008);
  \node[below] at (0.25,-0.008) {$\tfrac14$};
  \node[below] at (0.5,-0.008) {$\tfrac12$};
  \fill[purple!80!black] (0.4268,0.1768) circle (0.009) node[above right] {$m=+$};
  \fill[purple!80!black] (0.0732,-0.1768) circle (0.009) node[below right,xshift=2pt] {$m=-$};
  \draw[purple!80!black,densely dotted] (0.4268,0.1768) -- (0.4268,0);
  \draw[purple!80!black,densely dotted] (0.0732,-0.1768) -- (0.0732,0);
  \fill[orange!85!black] ($(0.1111,0)+(-0.009,-0.009)$) rectangle ($(0.1111,0)+(0.009,0.009)$);
  \fill[orange!85!black] ($(-0.1111,0)+(-0.009,-0.009)$) rectangle ($(-0.1111,0)+(0.009,0.009)$);
  \node[above,orange!85!black] at (0.1111,0.012) {$\times2$};
\end{tikzpicture}
\caption{The KD term $\KD(a,m)=\braket{m}{a}\bra{a}\hat\rho\ket{m}$ as the linear response of output $m$ to a weak object in path $a$, \cref{eq:response}. Its real part is the response to absorption, and a negative real part (shaded) marks an output that brightens under a weak absorber; its imaginary part is the response to a phase. Over all outputs the real parts sum to $P(a)$ and the imaginary parts sum to zero, \cref{eq:sumrules}. Circles: the two outputs of \cref{fig:fisher}(a)'s two-path interferometer, with input $(\ket b+e^{i\varphi}\ket a)/\sqrt2$ and balanced readout, at $\varphi=\pi/4$, where $\KD(a,\pm)=(1\pm e^{i\varphi})/4$; as $\varphi$ varies they move round the dashed circle, and their squared imaginary parts, divided by $P(m)$, are what \cref{thm:main} subtracts at $t=1$, \cref{eq:Fs1}. Squares: the three outputs of the three-path interferometer with the object on the internal path $\ket f=(\ket1+\ket2-\ket3)/\sqrt3$ and input $\ket{N_f}$ (\cref{sec:singleshot}), $\varrho(f,m)=(1,1,-1)/9$, real, and so optimal for an absorber, with one negative entry.}
\label{fig:kdplane}
\end{figure}

For an object that is not weak, \cref{eq:Pmt} can be written as
\begin{equation}
P(m|t)-P(m)=(1-\sqrt{t})\Big[(1-\sqrt{t})|\braket{m}{a}|^2P(a)-2\varrho(a,m)\Big],
\label{eq:brighten}
\end{equation}
so, for $t<1$, output $m$ is brighter with the object in place than without it if and only if
\begin{equation}
\varrho(a,m)<\tfrac12(1-\sqrt{t})\,|\braket{m}{a}|^2P(a),
\label{eq:brightcond}
\end{equation}
or, for a bright output, $(1-\sqrt{t})|\wv{m}|^2>2\,\mathrm{Re}[\wv{m}]$. A classical particle that is sometimes removed from path $a$ can make no output brighter, so \cref{eq:brightcond} marks the outputs whose statistics have no such explanation (\cref{sec:singleshot}). \cref{fig:setup}'s two cases are its ends. As $t\to1$ the condition becomes $\varrho(a,m)<0$, a negative KD term. A null port has $P(m)=0$ and therefore $\varrho(a,m)=0$, and satisfies \cref{eq:brightcond} for every $t<1$. Between them are outputs whose positive $\varrho(a,m)$ is small beside $|\braket{m}{a}|^2P(a)$, which darken under a weak absorber and brighten under a strong one. They are what a null port becomes when its extinction is imperfect. For real KD terms their weak value is large, $\wv{m}>2$, because $P(m)$ is small, below $\tfrac14|\braket{m}{a}|^2P(a)$, not because $\varrho(a,m)$ is negative.

\section{Fisher information from the unabsorbed photons}
\label{sec:fisher}

Suppose our task is to estimate $t$ from the photons the object does not absorb. This section bounds the information these photons hold, says which readouts reach the bound, and then drops two of the assumptions behind this bounding, a pure input and one object. \Cref{sec:helps} drops the other two assumptions: that there is no dispersion, and that the photon only makes one pass through the interferometer.

\subsection{The bound and its saturation}
\label{sec:bound}

The Fisher information about $t$ available from the unabsorbed photons, per photon entering the interferometer, is
\begin{equation}
\Fs(t)=\sum_m \frac{[\partial_t P(m|t)]^2}{P(m|t)},
\label{eq:Fsdef}
\end{equation}
with the convention that an output with $P(m|t)=0$ contributes the limit of its term in \cref{eq:Fsdef}, which is finite (see below). This is the information available when the incident flux is known but absorption events are not registered, for instance from the Poisson counts of a source with known mean. Writing $P(m|t)=S(t)q_m(t)$ with $\sum_m q_m=1$ decomposes this into $\Fs=P(a)^2/S(t)+S(t)F_{\mathrm{cond}}$, with the first term coming from the surviving flux and the second term from the conditional distribution of the detected photons. We consider the registered case below.

\begin{theorem}
\label{thm:main}
For every pure single-photon input state, every single-pass geometry, every output basis and every $t\in(0,1]$,
\begin{equation}
\begin{split}
\Fs(t)&=\frac{P(a)}{t}-\frac{1}{t}\sum_{m:P(m|t)>0}\frac{[\mathrm{Im}[\KD(a,m)]]^2}{P(m|t)} \\
&\le \frac{P(a)}{t},
\label{eq:thm}
\end{split}
\end{equation}
with equality if and only if every KD term $\KD(a,m)$ is real.
\end{theorem}

\begin{proof}
From \cref{eq:amp}, $\partial_t z_m=\braket{m}{a}\braket{a}{\psi}/(2\sqrt{t})$, so, separating the modulus and phase of $z_m$ as in Ref.~\cite[Eq.~(4)]{Braunstein1994} gives
\begin{equation}
\begin{split}
\frac{[\partial_t P(m|t)]^2}{P(m|t)}&=\frac{4\,[\mathrm{Re}(\bar z_m\partial_t z_m)]^2}{|z_m|^2}\\
&=4|\partial_t z_m|^2-\frac{4\,[\mathrm{Im}[\bar z_m\partial_t z_m]]^2}{|z_m|^2}.
\end{split}
\end{equation}
Using $\bar z_m\partial_t z_m=[\KD(a,m)-(1-\sqrt{t})|\braket{m}{a}|^2P(a)]/(2\sqrt{t})$, whose imaginary part is $\mathrm{Im}[\KD(a,m)]/(2\sqrt{t})$, and $\sum_m 4|\partial_t z_m|^2=\sum_m|\braket{m}{a}|^2P(a)/t=P(a)/t$ by completeness, \cref{eq:thm} follows. If $z_m(t)=0$ at some $t$ then $\KD(a,m)-(1-\sqrt{t})|\braket{m}{a}|^2P(a)=0$, so the subtracted term vanishes and the output contributes $4|\partial_t z_m|^2=|\braket{m}{a}|^2P(a)/t$, which is the finite limit referred to above. In particular $\KD(a,m)$ is then real, and $\KD(a,m)$ does not itself depend on $t$, so no output that is dark at some $t$ can bring an imaginary part into the subtracted sum at any other: that sum vanishes at one $t$ only if every $\KD(a,m)$ is real, which is the same condition at every $t$. An orthonormal basis is no restriction on the bound: for any POVM element $\hat E_m\ge0$, Cauchy--Schwarz for $\bra{x}\hat E_m\ket{y}$ gives $[\partial_tP(m|t)]^2/P(m|t)\le4\bra{\partial_tK_t\psi}\hat E_m\ket{\partial_tK_t\psi}$, and summing with $\sum_m\hat E_m=\mathds{1}$ gives $\Fs\le P(a)/t$ again. This step is Eq.~(11) of Ref.~\cite{Hofmann2011} for the generator $\hat\Pi(a)$ and the unnormalised state $K_t\hat\rho K_t$, since $\partial_tK_t=\hat\Pi(a)K_t/(2t)$. For a rank-one POVM $\hat E_m=\proj{v_m}$ the identity itself holds, with $\KD(a,m)=\braket{v_m}{a}\bra{a}\hat\rho\ket{v_m}$.
\end{proof}

The bound therefore holds for any POVM readout, and the identity for any rank-one POVM.

At $t=1$ the theorem gives
\begin{equation}
\begin{split}
\Fs(1)&=\sum_{m:P(m)>0}\frac{\varrho(a,m)^2}{P(m)}+\sum_{m:P(m)=0}|\braket{m}{a}|^2P(a)\\
&= P(a)-\sum_{m:P(m)>0}\frac{[\mathrm{Im}[\KD(a,m)]]^2}{P(m)},
\label{eq:Fs1}
\end{split}
\end{equation}
by the second sum rule in \cref{eq:sumrules}. The first form comes from inserting the linear response \cref{eq:linres}; the null ports contribute through the Elitzur--Vaidman term in \cref{eq:Pmt}, whose probability and squared derivative both vanish as $(1-t)^2$. The second form shows that negativity of $\varrho(a,m)$ is irrelevant to the total, which is fixed by the norm $\|\partial_t K_t\ket{\psi}\|^2$, so the only thing that can reduce the Fisher information below $P(a)$ is an imaginary part of the KD distribution. This form is the absorption counterpart of Eq.~(5) of Ref.~\cite{Hofmann2011}, in which the Fisher information about a phase is, up to a factor of four, the generator's variance less the variance of its real weak values: here the second moment $P(a)$ of $\hat\Pi(a)$ replaces the variance, and the real and imaginary parts exchange roles because $K_t=t^{\hat\Pi(a)/2}$ has a real exponent where a phase shift has an imaginary one. At $t<1$ the sign of $\varrho(a,m)$ does enter \cref{eq:thm}, through the occupation $P(m|t)$ that divides each subtracted term: an output that brightens dilutes the loss its own imaginary part causes, and one that darkens concentrates that loss. That changes how much a information a complex readout loses, but neither creates nor removes a loss, and the bound and its equality condition are unchanged. Imaginary parts are a property of the readout basis, not of the object: choosing the output basis so that all $\KD(a,m)$ are real restores the equality, and for a real input state every real readout is optimal. The phase of $\KD(a,m)$ is the angle, at output $m$, between the amplitude that arrived through path $a$ and the total amplitude there. When the two are in phase or in antiphase, attenuating the path changes the output's intensity at first order, which counting sees; when they are in quadrature it changes the output's phase instead, which counting does not see, and that share of the information is lost. \Cref{fig:fisher}(a) illustrates this for a two-path interferometer with a relative phase $\varphi$ between the arms and a fixed balanced readout. For $\varphi\in\{0,\pi\}$ the bound is saturated at every $t$, while any other $\varphi$ falls short of it by the ratio $1-(1+t)\sin^2\varphi/[(1+t)^2-4t\cos^2\varphi]$. This ratio tends to $\cos^2\varphi$ as $t\to0$ and is $1/2$ at $t=1$ and, for $|\cos\varphi|>\tfrac12$, dips below both to $1-(1+|\cos\varphi|)/(4|\cos\varphi|)$ at $t=2|\cos\varphi|-1$; the loss is recovered by rotating the readout by $\varphi$.

\begin{figure*}[t]
\centering
\panel{a}{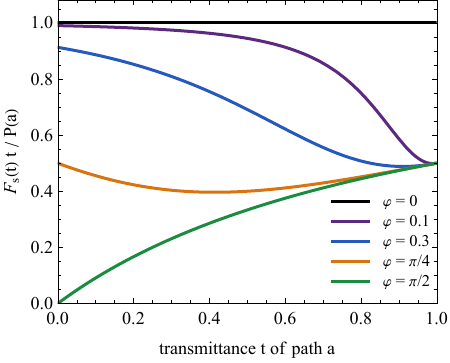}\hfill
\panel{b}{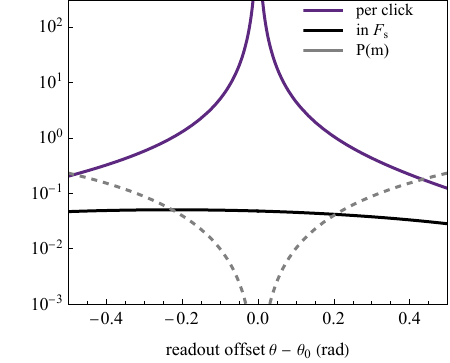}
\caption{(a) Fisher information from the unabsorbed photons, in units of the bound $P(a)/t$, against the transmittance $t$, for a two-path interferometer with input $(\ket{b}+e^{i\varphi}\ket{a})/\sqrt2$ and readout $(\ket{b}\pm\ket{a})/\sqrt2$; the curves are the values of $\varphi$ shown. The relative phase $\varphi$ makes the KD terms complex; the missing information is given by \cref{eq:thm}'s second term, and vanishes for a readout rotated by $\varphi$. (b) Weak-value amplification of an absorber, for $t=1$. Two paths with input $\sqrt{1-P(a)}\ket b+\sqrt{P(a)}\ket a$, $P(a)=0.05$, and readout $\ket{m_-}=\sin\theta\ket{b}-\cos\theta\ket{a}$, plotted on a logarithmic scale against the offset of $\theta$ from the angle $\theta_0=\arcsin\sqrt{P(a)}$ at which $\ket{m_-}$ is dark. The Fisher information per detected photon in the port, $(\mathrm{Re}[\wv{-}])^2$ (purple, ``per click''), diverges as the port darkens (at $t=1$ only; for $t<1$ it saturates for $\theta>\theta_0$, where $\mathrm{Re}[\wv{-}]<0$, \cref{sec:bound}), while its contribution to $\Fs$, $P(m_-)(\mathrm{Re}[\wv{-}])^2=P(a)\cos^2\theta$ (black, ``in $\Fs$''), never exceeds $P(a)$; the port occupancy $P(m_-)$ is the dashed grey curve.}
\label{fig:fisher}
\end{figure*}

\begin{corollary}[per absorbed photon]
\label{cor:absorbed}
The object absorbs $(1-t)P(a)$ photons per trial, so $\Fs(t)/[(1-t)P(a)]\le 1/[t(1-t)]$. Sending the whole photon through the object and detecting it ($P(a)=1$, one output) gives $\Fs=1/t$ and $(1-t)$ absorbed photons per trial, hence information-per-photon $1/[t(1-t)]$. No single-pass interferometer, single-photon input or readout extracts more Fisher information per absorbed photon from the unabsorbed photons than direct transmission (for a mixed input, by \cref{eq:mixed}), and every configuration with real KD terms and a pure input extracts as much, null ports included.
\end{corollary}

If the photon number is known, such that the absence of a click is itself an outcome of probability $(1-t)P(a)$ with derivative $-P(a)$, the total Fisher information is $\Fs(t)+P(a)/(1-t)\le P(a)/[t(1-t)]$, i.e., at most $1/[t(1-t)^2]$ per absorbed photon. This is the single-pass case of Wildfeuer's bound \cite[Eq.~(25)]{Wildfeuer2026}, and of the earlier bounds of Ref.~\cite{Massar2001}, for any protocol distinguishing close transparencies, and Ref.~\cite{Facchi2002}, for binomial statistics; the factor $1/(1-t)$ between the two cases is the $\beta^2=1-t$ that separates Ref.~\cite[Sec.~II]{Massar2001}'s absorbed-photon prefactors, $\beta^4$ for a known photon number against $\beta^2$ for coherent light, which must therefore absorb $1/(1-t)$ times as many photons for the same resolution. This argument also shows that every single-pass readout with real KD terms saturates the bound for a pure input (mixed inputs are the subject of \cref{sec:extensions}), which is why the Elitzur--Vaidman interferometer and direct transmission give the same Fisher information per absorbed photon \cite{Wildfeuer2026}. \cref{eq:Fs1} identifies the mechanism behind this. Detectors with the same efficiency change neither conclusion (see Appendix~\ref{app:efficiency}).

\begin{corollary}[per detected photon]
\label{cor:detected}
The Fisher information per photon detected in output $m$ is $[\partial_t\ln P(m|t)]^2$, which at $t=1$ equals $(\mathrm{Re}[\wv{m}])^2$, the square of the real part of the weak value of the path projector. Classically $\mathrm{Re}[\wv{m}]\in[0,1]$, so this is at most one; a weak value with $|\mathrm{Re}[\wv{m}]|>1$, which is necessarily anomalous, gives more than one unit of Fisher information per detected photon in that port.
\end{corollary}

This is weak-value amplification of an absorber. The information is not created but moved: by \cref{eq:Fs1} output $m$ contributes $\varrho(a,m)^2/P(m)=P(m)(\mathrm{Re}[\wv{m}])^2$ to $\Fs(1)$, which is at most $|\braket{m}{a}|^2P(a)$ because $|\KD(a,m)|^2=P(m)|\braket{m}{a}|^2P(a)$, and these contributions sum to at most $P(a)$. A port giving more than one unit per detected photon must therefore be a weak port, $P(m)<|\braket{m}{a}|^2P(a)\le P(a)$: the amplification is a smaller share of the photons, not a larger share of the information. For a negative weak value it also needs a weak absorber. For real KD terms, \cref{eq:LR} gives the information per detected photon at $t<1$ as $w^2/[t(1-(1-\sqrt{t})w)^2]$, with $w=\wv{m}$. For negative $w$ this grows with $|w|$ only until $(1-\sqrt{t})|w|$ is of order one, and never exceeds $1/[t(1-\sqrt{t})^2]$, its value at a port that is dark without the object, so the gain per detected photon over direct transmission is at most $t/(1-\sqrt{t})^2$, set by the absorber and not by how dark the port is. \Cref{fig:fisher}(b) shows the two-path case, where the per-photon information in a nearly dark port diverges as $1/\sin^2(\theta-\theta_0)$ while its contribution to $\Fs$ is $P(a)\cos^2\theta$. This is the trade-off analysed for weak-value amplification of unitary parameters, where postselection also cannot raise the information per trial \cite{Ferrie2014,Combes2014}; for a phase in a two-path interferometer, Fukuda \textit{et al} have concentrated the information per detected photon in the dark port to the level of a four-photon NOON fringe by a small path--polarisation entanglement, which required negative KD terms \cite{Fukuda2026}. Whether concentrating the information into few detected particles helps in practice is disputed.\footnote{Jordan \textit{et al} argue it can help, since a weak port can hold as much Fisher information as the full beam while the discarded light is free for another experiment or for recycling, and some technical noises are suppressed \cite{Jordan2014}; Knee and Gauger find no advantage against detector pixelation or jitter \cite{KneeGauger2014}.} For an absorber, Ref.~\cite{Ji2026}'s generator is the path projector, and \cref{cor:detected}'s real part of the weak value is its conditional average, Eq.~(6) there, anomalous when it leaves the eigenvalue range $[0,1]$; the sum rule behind \cref{eq:Fs1} is the absorption counterpart of their decomposition of the variance (Appendix~\ref{app:efficiency}).

For comparison with the interferometer-free bounds we list in the Introduction, a coherent state of mean photon number $N$ sent into the same optics leaves every output in a coherent state of amplitude $\sqrt N z_m(t)$, and the quantum Fisher information of that family is $4N\sum_m|\partial_tz_m|^2=NP(a)/t$, which for $P(a)=1$ is the coherent-state value $N/t$ \cite{Monras2007}. Photon counting at the outputs of coherent light gives $N\Fs(t)$, so \cref{thm:main} says which readouts reach the coherent-state limit: those with real KD terms. By Nair's bound the registered bound holds for every multi-photon input to a single-pass interferometer, entangled across the paths or not (Appendix~\ref{app:efficiency}).

\subsection{Mixed inputs and several objects}
\label{sec:extensions}

Neither a mixed input nor several absorbers gives the interferometer any informational benefit over direct transmission.

\textit{Mixed inputs.} \Cref{eq:Pmt,eq:linres} are linear in $\hat\rho$ and hold unchanged for a mixed input, with $\KD(a,m)=\braket{m}{a}\bra{a}\hat\rho\ket{m}$ and $P(a)=\bra{a}\hat\rho\ket{a}$. Writing $\partial_tP(m|t)=\mathrm{Re}[x_m]/\sqrt{t}$ with $x_m=\braket{m}{a}\bra{a}\hat\rho K_t\ket{m}$, whose imaginary part is again $\mathrm{Im}[\KD(a,m)]$, and using $(\mathrm{Re}[x_m])^2=|x_m|^2-(\mathrm{Im}[x_m])^2$ as in the proof of \cref{thm:main},
\begin{equation}
\begin{split}
&\Fs(t)=\frac{P(a)}{t}-\frac1t\sum_{m:\,P(m|t)>0}\frac{[\mathrm{Im}[\KD(a,m)]]^2}{P(m|t)}\\
&-\frac1t\sum_{m:\,P(m|t)>0}|\braket{m}{a}|^2\bigg[P(a)-\frac{|\bra{a}\hat\rho K_t\ket{m}|^2}{P(m|t)}\bigg].
\label{eq:mixed}
\end{split}
\end{equation}
The last sum is new. Each of its terms is non-negative by the Cauchy--Schwarz inequality, $|\bra{a}\hat\rho K_t\ket{m}|^2\le\bra{a}\hat\rho\ket{a}\bra{m}K_t\hat\rho K_t\ket{m}$, and each vanishes for a pure state. The bound $\Fs\le P(a)/t$ therefore holds for every input, but real KD terms no longer guarantee equality: the readout must also not rely on coherence that the input lacks. Reading out in the path basis always does this, since then only $m=a$ has $\braket{m}{a}\neq0$ and $K_t\ket{a}=\sqrt{t}\ket{a}$; this is direct transmission of the fraction $P(a)$ of the light.

The readouts that lose most are those with ports the input's coherence would leave dark. Take the balanced Elitzur--Vaidman interferometer with the coherence between its arms reduced by a visibility $V$, $\hat\rho=\tfrac12(\proj{a}+\proj{b})+\tfrac V2(\ketbra{a}{b}+\ketbra{b}{a})$. Then $P(\pm|t)=(1+t\pm2V\sqrt{t})/4$, the bright port contributes $(1+V/\sqrt{t})^2/[4(1+t+2V\sqrt{t})]$ to $\Fs$, and the dark port contributes
\begin{equation}
\frac{[(1-\sqrt{t})-(1-V)]^2}{4t\,[(1-\sqrt{t})^2+2\sqrt{t}(1-V)]} .
\label{eq:visibility}
\end{equation}
At $V=1$ each port gives $1/(4t)$ and their sum is $P(a)/t$. For any $V<1$ the dark-port term tends to $(1-V)/8$ as $t\to1$ and the total to $1/4=P(a)/2$: half the information is lost, however good the visibility, once the object is weak enough that $(1-\sqrt{t})^2\ll2(1-V)$. The reason is that the null port holds its share of the information in a probability of order $(1-\sqrt{t})^2$, which any background in that port swamps. Half is the balanced interferometer's share. More generally (Appendix~\ref{app:mixed}), each port keeps of its share of path $a$ only the fraction of its counts that comes from the part of the input coherent with that path. A port that this part would leave dark, or nearly so, forfeits its share once a background fills it, while an ideal null port, dark for the whole input, loses nothing. The pure-input qualifier in \cref{cor:absorbed} is therefore needed.

\textit{Several objects.} Take a pure input again, $\hat\rho=\proj\psi$, and let every path $j$ contain an object of transmittance $t_j$ ($t_j=1$ for an empty path), so that $z_m=\sum_j\sqrt{t_j}\braket{m}{j}\braket{j}{\psi}$. The Fisher information matrix of the unabsorbed photons is $F_{jk}=\sum_m\partial_{t_j}P(m|\mathbf t)\,\partial_{t_k}P(m|\mathbf t)/P(m|\mathbf t)$, at the given transmittances $\mathbf t$. Completeness of the output basis now gives $\sum_m\overline{\partial_{t_j}z_m}\,\partial_{t_k}z_m=\delta_{jk}P(j)/(4t_j)$: the derivatives with respect to different paths are orthogonal. The steps of the proof of \cref{thm:main} then give, at any $\mathbf t$ at which every output is occupied (the term of an empty output has no single limit when several transmittances vary),
\begin{equation}
F_{jk}=\frac{P(j)}{t_j}\,\delta_{jk}-4\sum_m\frac{\mathrm{Im}[\bar z_m\partial_{t_j}z_m]\,\mathrm{Im}[\bar z_m\partial_{t_k}z_m]}{P(m|\mathbf t)},
\label{eq:fim}
\end{equation}
and the subtracted matrix is positive semidefinite, so $F\le\mathrm{diag}[P(j)/t_j]$ as a matrix inequality: no linear combination of the transmittances is estimated better than by sending the fraction $P(j)$ of the light straight through each object. For a real input state and a real readout, a real orthogonal mixing of the paths such as a Hadamard transform, the subtracted matrix vanishes: every $t_j$ is then estimated at its direct-transmission value at once, and the estimates are uncorrelated. Appendix~\ref{app:several} explains the restriction on $\mathbf t$ and gives the sum rule over the KD terms of all paths that lies behind \cref{eq:fim} for weak objects. With the absorbed photons registered the corresponding statement is Nair's bound on the quantum Fisher information matrix of several loss elements probed in parallel, $4\,\mathrm{diag}(N_k)$ in the angles $\cos\theta_k=\sqrt{t_k}$ and $\mathrm{diag}\{N_k/[t_k(1-t_k)]\}$ in the transmittances \cite{Nair2018}.

\section{When interference helps}
\label{sec:helps}

\Cref{thm:main} leaves two ways an interferometer can increase the information per absorbed photon, which we discuss here. The first is that a dispersive object adds a phase response, which can turn the imaginary parts of the KD terms into signal rather than loss (\cref{sec:dispersive}). The second is that many passes add interference between the responses of different passes, which can allow us to gain information, but only for the case where the absorptions go unregistered (\cref{sec:multipass}).

\subsection{A dispersive object}
\label{sec:dispersive}

The two kinds of object (absorbers, and weak phase objects), ask opposite things of the readout. An absorber is best read out with real weak values, whatever their spread. A weak phase object is best read out with $\mathrm{Re}[\wv{m}]=P(a)$ at every output, so that absorption changes the outputs only through the survival probability and no real part is anomalous. The information is then held in the imaginary parts, which only subtract from the information about an absorber but are the whole signal for a phase. Dispersion is the one assumption of \cref{thm:main} whose relaxation lets a single pass beat direct transmission, and the phase term is unchanged when the absorbed photons are registered.

With $g=\sqrt{t} e^{i\phi}$ (as in \cref{eq:response}), let both $t$ and $\phi$ depend on the quantity $c$ to be estimated, a concentration or a thickness, with primes denoting derivatives with respect to $c$. The survival probability $S$ is unchanged. Now $\partial_cz_m=g'\braket{m}{a}\braket{a}{\psi}$, and the first step of the proof of \cref{thm:main} gives
\begin{equation}
\Fs(c)=\Big(\frac{t'^2}{t}+4t\phi'^2\Big)P(a)-4\sum_{m:P(m|c)>0}\frac{[\mathrm{Im}[\bar z_m\partial_cz_m]]^2}{P(m|c)} .
\label{eq:Fsc}
\end{equation}
The difference from pure absorption is that the subtracted term can no longer be made to vanish, because $\sum_m\mathrm{Im}[\bar z_m\partial_cz_m]=\mathrm{Im}[\bra{\psi}K^\dagger\partial_cK\ket{\psi}]=t\phi'P(a)$ is fixed by the object.

\begin{theorem}
\label{thm:dispersive}
For every pure single-photon input state, every single-pass geometry, every output basis and every $c$ with $t(c)\in(0,1]$,
\begin{equation}
\Fs(c)\le P(a)\Big[\frac{t'^2}{t}+4t\phi'^2\,\frac{1-P(a)}{S}\Big],
\label{eq:thmdisp}
\end{equation}
with equality if and only if $\mathrm{Im}[\partial_c\ln z_m]$ takes the same value for every output with $P(m|c)>0$. An output basis with this property exists at every $c$.
\end{theorem}

\begin{proof}
By the Cauchy--Schwarz inequality,
\begin{equation}
\begin{split}
\sum_m\frac{[\mathrm{Im}[\bar z_m\partial_cz_m)]]^2}{P(m|c)}&\ge\frac{\big[\sum_m\mathrm{Im}[\bar z_m\partial_cz_m]\big]^2}{\sum_mP(m|c)}\\
&=\frac{t^2\phi'^2P(a)^2}{S},
\end{split}
\end{equation}
with equality if and only if $\mathrm{Im}[\bar z_m\partial_cz_m]/P(m|c)=\mathrm{Im}[\partial_c\ln z_m]$ does not depend on $m$. Inserting this into \cref{eq:Fsc}, and using $1-tP(a)/S=[1-P(a)]/S$ gives \cref{eq:thmdisp}. For attainability, let us write $\Fs=(\partial_cS)^2/S+S\,F_{\mathrm{cond}}$ as in \cref{sec:bound}. The right-hand side of \cref{eq:thmdisp} equals $(\partial_cS)^2/S+S\mathcal Q$, where $\mathcal Q$ is the quantum Fisher information of the normalised surviving state $\ket{e_1}=K\ket\psi/\sqrt S$, and a projective measurement with $F_{\mathrm{cond}}=\mathcal Q$ exists for every pure-state family \cite{Braunstein1994}; here it is the eigenbasis of the symmetric logarithmic derivative, a basis containing $(\ket{e_1}\pm\ket{e_2})/\sqrt2$, where $\ket{e_2}$ is the component of $\partial_c\ket{e_1}$ orthogonal to $\ket{e_1}$, divided by its norm.
\end{proof}

For $\phi'=0$, the typical value of $\mathrm{Im}[\partial_c\ln z_m]$ is zero, and \cref{thm:dispersive} is \cref{thm:main}. The optimal readout in general depends on $c$, so reaching the bound needs a prior estimate or an adaptive readout, as for any phase. The content is clearest for a weak object, $g=1$. There $\bar z_m\partial_cz_m=g'\,\KD(a,m)$ with $g'=t'/2+i\phi'$, so that, as \cref{eq:response} says,
\begin{equation}
\partial_cP(m)=t'\,\varrho(a,m)-2\phi'\,\mathrm{Im}[\KD(a,m)].
\label{eq:complexresponse}
\end{equation}
The Fisher information is $\sum_mP(m)\left(t'\,\mathrm{Re}[\wv{m}]-2\phi'\,\mathrm{Im}[\wv{m}]\right)^2$, and the equality condition of \cref{thm:dispersive} becomes $\tfrac12t'\,\mathrm{Im}[\wv{m}]+\phi'\,\mathrm{Re}[\wv{m}]=\phi'P(a)$ for every $m$ with $P(m)>0$. Setting $\phi'=0$ and $t'=0$ in turn gives the two readouts contrasted above.

Per absorbed photon, \cref{eq:thmdisp} is
\begin{equation}
\frac{\Fs(c)}{(1-t)P(a)}\le\frac{t'^2}{t(1-t)}+\frac{4t\phi'^2}{1-t}\,\frac{1-P(a)}{S} .
\label{eq:disp-abs}
\end{equation}
The first term is direct transmission, which no interferometer improves on. The second term is absent for direct transmission, which has $P(a)=1$ and is blind to $\phi$; it is largest when the interferometer sends little light through the object, $P(a)\to0$, where it is $4t\phi'^2/(1-t)$. (With the absorbed photons registered the first term becomes $t'^2/[t(1-t)^2]$ and the second is unchanged.) For a line of Lorentzian shape, $g=\exp[-c\kappa/(\gamma-i\delta)]$ with detuning $\delta$ and half-width $\gamma$, the ratio of the second term to the first as $P(a)\to0$ is $(\delta/\gamma)^2$, the same quadratic growth with detuning that Wildfeuer found for phase sensing per absorbed photon \cite{Wildfeuer2026b}. More than a half-width from resonance, most of the information per absorbed photon is in the term that direct transmission cannot reach. This explains, already for a single pass, Ref.~\cite{Wildfeuer2026}'s remark that a Kramers--Kronig phase, building up over passes, takes the object outside the absorption bounds. It does not conflict with \cref{cor:absorbed}: the additional information is about $\phi$, it comes from the imaginary parts of the KD terms, and negativity of the real parts plays no role in it. The same trade appears in the phase-and-loss estimation literature, where the loss is a known channel parameter and the phase the one estimated \cite{Escher2011,Demkowicz2012}, and where the trade-off in estimating both at once has been analysed \cite{Crowley2014}; here they are one parameter.

For a Lorentzian line the ratio $(\delta/\gamma)^2$ compares the two terms at the same detuning; for a weak sample and $P(a)\to0$ they sum to the resonant direct-transmission value at every detuning, so a gain over probing on resonance needs either an optically thick sample, whose direct transmission is limited by its optical depth, or a line broadened beyond the natural width that sets its wing. A Lorentzian sample of resonant optical depth $x_0$ then gains up to $1.54\,x_0$ per absorbed photon over the best direct transmission with the absorptions unregistered (for $x_0\ge1.59$) and up to $1+x_0$ with them registered, both as $P(a)\to0$, and, with them unregistered, a room-temperature caesium vapour cell probed in the Lorentzian wing of its Doppler-broadened line gains about $50$ times the resonant direct-transmission value of a weak sample, whereas a molecular band in solution offers little phase gain and a sample inside a ring resonator none of this kind (Appendix~\ref{app:dispersive}).

\subsection{Many passes}
\label{sec:multipass}

For an object that only absorbs, no fixed arrangement of one photon beats direct transmission once the absorbed photons are registered \cite{Wildfeuer2026}, and Wildfeuer lists absorption estimation with counted losses as a flat exchange rate, the same per absorbed photon for any number of passes \cite{Wildfeuer2026b}, which makes direct transmission the right baseline. A multi-pass scheme of one photon can therefore gain information only in the gap between the registered and unregistered bounds. \Cref{thm:main} assumes of the geometry only that the photon meets the object once; dropping that assumption leaves the identity behind it intact and shows where the gain can come from.

\begin{proposition}
\label{prop:manypass}
Let the photon meet the object $n$ times with arbitrary lossless optics between passes, so that the unabsorbed state is $\ket{\Psi_t}=U_nK_t\cdots U_1K_t\ket\psi$. Write $\ket{\Psi^{(k)}_t}=U_kK_t\ket{\Psi^{(k-1)}_t}$ with $\ket{\Psi^{(0)}_t}=\ket\psi$, let $c_k=\braket{a}{\Psi^{(k-1)}_t}$ be the amplitude on the object's path just before the $k$th pass, and let $A_k=U_nK_t\cdots U_{k+1}K_tU_k$ be the propagator from just after the object on that pass to the output, so $A_n=U_n$. Then for every output basis and every $t\in(0,1]$,
\begin{equation}
\Fs(t)=\frac{1}{t}\Big\|\sum_{k=1}^n c_k A_k\ket{a}\Big\|^2-4\hspace{-6pt}\sum_{m:P(m|t)>0}\frac{[\mathrm{Im}[\bar z_m\partial_t z_m]]^2}{P(m|t)},
\label{eq:manypass}
\end{equation}
and the object absorbs $(1-t)\sum_k|c_k|^2$ photons per trial.
\end{proposition}

\begin{proof}
The first step of the proof of \cref{thm:main} uses only completeness of the output basis, so for any family of amplitudes $\Fs=4\sum_m|\partial_t z_m|^2-4\sum_m[\mathrm{Im}[\bar z_m\partial_t z_m]]^2/P(m|t)$, the first sum being $4\|\partial_t\Psi_t\|^2$ and an output with $z_m=0$ contributing $4|\partial_tz_m|^2$ to it and nothing to the second. Differentiating the product and using $\partial_tK_t=\proj a/(2\sqrt{t})$ gives $\partial_t\ket{\Psi_t}=\sum_kc_kA_k\ket a/(2\sqrt{t})$. The absorption follows from $\|K_t\ket\phi\|^2=\|\ket\phi\|^2-(1-t)|\braket{a}{\phi}|^2$ applied pass by pass.
\end{proof}

\Cref{thm:main} is \cref{eq:manypass}'s $n=1$ case, where $A_1$ is unitary and the norm is $P(a)/t$ whatever the optics. For more than one pass the norm is a coherent sum, and it departs from $\sum_k|c_k|^2\|A_k\ket a\|^2/t$, the value the passes would give independently, only through the overlaps $\bra{a}A_j^\dagger A_k\ket{a}$ with $j\neq k$, which lower it as readily as raise it. Those overlaps are the only route to a gain: each $\|A_k\ket a\|\le1$, so without them $\Fs\le\sum_k|c_k|^2/t$, which against the absorption $(1-t)\sum_k|c_k|^2$ is at most the direct-transmission value once more. The chain below is that interference arranged deliberately.

Thomas \textit{et al} found numerically that Refs.~\cite{Kwiat1995,Kwiat1999}'s Zeno interferometer, read out with a Poisson source and unregistered loss, needs fewer absorbed particles than direct transmission for $t$ above about one half; for a suitable number of cycles it needs more than an order of magnitude fewer for $t\gtrsim0.95$ (their Fig.~9, with $10$, $100$ and $500$ cycles) \cite{Thomas2014}. The multi-pass dose advantage has been measured in a self-imaging cavity, in both polarisation and transmission contrast, where losses in the sample and the cavity (in the polarisation data chiefly the cavity's $19\%$ per pass) stopped the gain growing after three to seven passes \cite{Juffmann2016}. The model used below has $n$ cycles, each rotating the photon by $\pi/(2n)$ from the reference arm into the sample arm, whose amplitude is then multiplied by $\sqrt t$; the two arms are detected after the last cycle, and the loss again goes unregistered. Its crossover, the transmittance above which it beats direct transmission, rises with the number of cycles (Appendix~\ref{app:zeno}). The advantage comes from many coherent passes converting the loss into a change of output port (Appendix~\ref{app:zeno} gives the numbers); \cref{thm:main} is the statement that a single pass cannot do this, whichever interferometer and readout are used.

The chain's numbers (Appendix~\ref{app:zeno}) are close to the best the model allows. In the limit $t\to1$ at fixed $\nu=n(1-\sqrt t)$ the chain becomes a damped two-level system, and the fraction of the registered bound that the two output ports deliver becomes a function $G(\nu)$ of that one variable (\cref{fig:zeno,app:zeno}). Its maximum is $G^*=0.932$, and no readout of the two ports does better; registering the absorption would add only $0.003$ of the bound. The missing $0.065$ is the record of which cycle absorbed the photon, which stays with the object and which no readout of the photon recovers \cite{Wildfeuer2026}. Light passed $\ell$ times straight through the object with its loss unregistered reaches at most $0.648$ of the registered bound in the same limit (Appendix~\ref{app:zeno}), the $65\%$ Wildfeuer finds at $t=0.99$ and $\ell=159$ \cite{Wildfeuer2026}.

The chain is not the best arrangement: an all-pass ring resonator with the object inside it brings $G$ arbitrarily close to one on resonance. It is another instance of \cref{prop:manypass}, the photon meeting the object once per round trip. With self-coupling amplitude $r_{\mathrm c}$ at the coupler, round-trip amplitude transmission $\sqrt{t}$ set by the object, round-trip phase $\theta_{\mathrm{rt}}$ and a wavepacket long compared with the ring's lifetime, so that the ring is in its steady state, the output amplitude is $g_{\mathrm{ring}}=(r_{\mathrm c}-\sqrt{t} e^{i\theta_{\mathrm{rt}}})/(1-r_{\mathrm c}\sqrt{t} e^{i\theta_{\mathrm{rt}}})$ \cite{Bogaerts2012}. The buildup, the circulating intensity per photon entering as it reaches the object, is $\mathcal B=(1-r_{\mathrm c}^2)/|1-r_{\mathrm c}\sqrt{t} e^{i\theta_{\mathrm{rt}}}|^2$ (Ref.~\cite{Belsley2022} takes it after the round trip, where it is $t\mathcal B$), and the object absorbs $(1-t)\mathcal B$ per photon entering. On resonance the amplitudes are real, the second term of \cref{eq:manypass} vanishes, and the first is $\mathcal B^2/t$: the coherent sum over passes is the buildup squared. The output port alone therefore gives $\mathcal B/[t(1-t)]$ per absorbed photon, $\mathcal B$ times the direct-transmission value, and the fraction of the registered bound it delivers is $G=(1-t)\mathcal B$, the absorbed fraction itself; the absorbed outcome supplies the rest, so that the two together attain the bound at every coupling. Since $\mathcal B\le1/(1-t)$, with equality at critical coupling ($r_{\mathrm c}=\sqrt{t}$), $G\le1$, and $G=1$ at critical coupling on resonance: the output port is then dark, the whole photon is absorbed, the absorption is stationary in $t$, and the dark port gives $1/[t(1-t)^2]$ per absorbed photon by itself, through the finite limit of \cref{eq:Fsdef}. Because $r_{\mathrm c}=\sqrt{t}$ depends on the unknown $t$, reaching it needs a prior estimate or an adaptive coupler, as the optimal readout of a dispersive object does (\cref{sec:dispersive}), and the estimate must also say on which side of critical coupling $t$ lies, since the output intensity vanishes quadratically there and values of $t$ on either side give the same counts. This is Ref.~\cite{Belsley2022}'s setting, whose authors note that the output is dark at critical coupling and that every photon is then absorbed, but count per photon entering rather than per photon absorbed. Off resonance the amplitudes are complex and $G<1$, and at critical coupling discontinuously so: a detuning fills the dark port with an amplitude in quadrature with its response to $t$, and the second term of \cref{eq:manypass} then removes nearly all of the first, so that for small detunings $G\simeq(1+t)^2\theta_{\mathrm{rt}}^2/[4(1-t)^2]$, about the square of the detuning in half-widths of the resonance when $t$ is near one (Appendix~\ref{app:zeno}). The shortfall belongs to the readout, counting photons at the output port, not to the output field, whose information, read against a phase reference, changes with the detuning only through $\mathcal B$ \cite{Belsley2022}.

A wavepacket of $N$ round trips through a critically coupled ring on resonance gives, at $t=0.9$ and for $N=40$, $80$ and $1000$, $G=0.91$, $0.96$ and $0.998$ if the light left in the ring at the end is detected as a further output, and $0.86$, $0.93$ and $0.995$ if the ring is left to empty through its coupler, short of one only by the transients at the ends of the wavepacket. For every wavepacket of finite length $G<1$, as for every arrangement with finitely many passes: $G=1$ needs both the absorbed outcome and the record of which pass absorbed the photon, which the object keeps \cite{Wildfeuer2026}, to hold no information, and the record holds none only if every pass's share of the absorption is stationary in $t$, and then, since the first pass's amplitude does not depend on $t$, the absorption itself is not stationary (Appendix~\ref{app:zeno}). Only the steady state (reached in the limit of a long wavepacket, where the absorption is uniform in time), escapes this. These values count each output time bin as a separate outcome. The ring's information is held in a port that is dark at critical coupling, so anything that lights that port lowers $G$. Counted over the whole wavepacket, with the ring left to empty, the bright transients at its ends do so, and the three wavepackets above then give $G=0.067$, $0.034$ and $0.0024$, below direct transmission's $1-t=0.1$. Moving the coupling off critical, by $0.10$, $0.08$ and $0.04$ below $\sqrt{t}$, restores $G=0.56$, $0.65$ and $0.86$, since the steady state then lights the through port as well. A detector background, a source broader than a small fraction of the ring's linewidth, or loss in the ring other than the object's each lower it too (Appendix~\ref{app:zeno}).

\begin{figure}[t]
\centering
\includegraphics[width=\columnwidth]{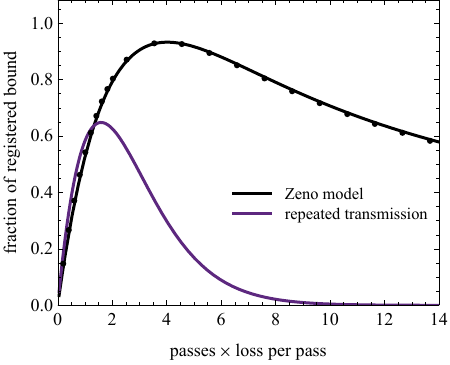}
\caption{Many passes with unregistered loss: Fisher information per absorbed photon delivered by the detected photons, as a fraction of the registered bound $1/[t(1-t)^2]$, against the number of passes times the intensity loss per pass, in the limit $t\to1$. Black: \cref{sec:multipass}'s Zeno model, Appendix~\ref{app:zeno}'s $G(\nu)$ plotted against $2\nu=2n(1-\sqrt t)$, which equals $n(1-t)$ in this limit and so shares the abscissa with the purple curve; its maximum $0.932$ at $2\nu=4.03$ is Appendix~\ref{app:zeno}'s $\nu^*=2.013$. Points: the $n$-cycle model at $t=0.95$, against $2n(1-\sqrt t)$. Purple: light passed $\ell$ times straight through the object, $y^2/(e^y-1)$ with $y=\ell(1-t)$, with maximum $0.648$ at $y=1.59$ \cite{Wildfeuer2026,Wildfeuer2026b}.}
\label{fig:zeno}
\end{figure}

\section{Discrimination: single-shot gain that does not accumulate}
\label{sec:discrimination}

We now consider the task interaction-free measurement was introduced for: determining whether an object is present or absent. \Cref{sec:singleshot} follows Ref.~\cite{Hance2024}, with the photon heralded, meaning an absorption is itself an observed outcome; \cref{sec:manyshot} gives the error exponent for both that case and \cref{sec:fisher}'s unregistered case; and \cref{sec:optimalreadout} identifies the optimal single-pass readout. (Appendices~\ref{app:finiteM} discusses how long the single-shot gain lasts, and \ref{app:sequential} considers sequential tests for the unregistered case.) For a weak absorber, the conclusion about negativity is the same for both the registered and unregistered case; however, their conclusions for null ports differ, since they gain information at leading order in $1-t$ only in the unregistered case.

\subsection{One photon}
\label{sec:singleshot}

Anisimov \textit{et al} posed interaction-free measurement as a symmetric test between a blocked and a clear path, with the absorbed outcome included in the Chernoff exponent, and measured the scheme by the ratio of that exponent to the rate at which the object absorbs \cite{Anisimov2010}: \cref{fig:disc}(b)'s ratio $C(t)/[-\ln S(t)]$. Zhou and Yung formulated interaction-free measurement as channel discrimination for a semitransparent object: at a finite number of Zeno cycles they minimised the loss probability of a single photon and its error probability separately over the input state, and compared the loss of the zero-error inputs with the least attainable loss \cite{Zhou2017}. Ref.~\cite{Hance2024} quantified the benefits of interaction-free measurement, for an opaque object, by the average error of a single-shot guess between object present and object absent with equal priors. That error is $P_{\mathrm{err}}=(1-\Delta_a$)/2, where $\Delta_a$ is the statistical distance between the outcome distributions, the absorbed outcome included. For a partially transmitting object,
\begin{figure*}[t]
\centering
\panel{a}{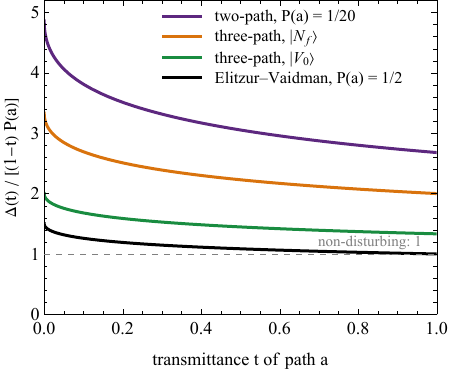}\hfill
\panel{b}{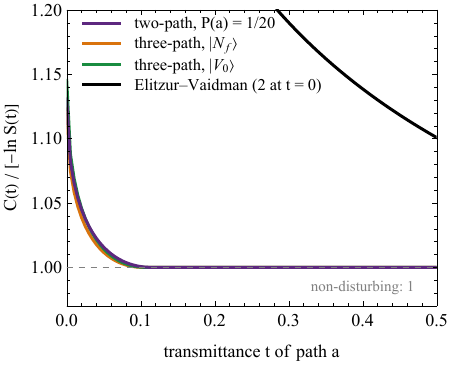}
\caption{Discrimination between object present (transmittance $t$) and absent, per absorbed photon, for four interferometers: two paths with $P(a)=1/20$ and balanced readout (purple), the three-path interferometer with inputs $\ket{N_f}$ (orange) \cite{Hofmann2023,Sagawa2026} and $\ket{V_0}$ (green) \cite{Sagawa2026}, and the balanced Elitzur--Vaidman interferometer (black). (a) Single-shot statistical distance $R(t)$, \cref{eq:R}, over $t\in[0,1]$; the $t\to1$ limits are $1+\Nneg/P(a)$, with $\Nneg$ the negativity of the KD distribution. (b) Error exponent of a repeated test, the Chernoff information $C(t)$, relative to its value $-\ln S(t)$, where $S(t)$ is the survival probability, for any non-disturbing definite-path scheme with the same absorption, over $t\in[0,0.5]$ and with the ratio axis cut at $1.2$. Panel (a)'s negativity-driven gain is absent from panel (b) for $t$ above $t^*\simeq0.10$--$0.11$, where $C(t)$ and $-\ln S(t)$ coincide. The Elitzur--Vaidman interferometer has $C(t)=-2\ln[(1+\sqrt{t})/2]$ from its null port; the ratio is $2$ at $t=0$, falls below $1.2$ (and so into the panel) near $t=0.3$, and reaches $1$ only at $t=1$, beyond the range shown.}
\label{fig:disc}
\end{figure*}
\begin{equation}
\Delta_a(t)=\frac12\Big[(1-t)P(a)+\sum_m\big|P(m|t)-P(m)\big|\Big].
\label{eq:tv}
\end{equation}
Let us call a model \emph{non-disturbing definite-path}, or non-disturbing for short, if the photon has a definite path and the object's only effect on the photons that survive is to delete those that were on path $a$, so that no output brightens: $\varrho(a,m)\ge\tfrac12(1-\sqrt{t})|\braket{m}{a}|^2P(a)$ at every $m$, which is the negation of \cref{eq:brightcond} and becomes $\varrho(a,m)\ge0$ as $t\to1$. In any such model, $\sum_m|P(m|t)-P(m)|=(1-t)P(a)$ and $\Delta_a=(1-t)P(a)$: the object is detected only through the photons it absorbs. (This is not \cref{sec:setting}'s noncontextuality, which is generalised noncontextuality in the sense of Refs.~\cite{Pusey2014,Lostaglio2020} and constrains measurement statistics rather than forbidding disturbance. The Elitzur--Vaidman null port violates the non-disturbance condition at every $t<1$. However, for an opaque object, its statistics have been reproduced in models where blocking a path also randomises a hidden variable: in a balanced interferometer by Ref.~\cite{Catani2023}'s toy field theory, and with unbalanced beam splitters too in Ref.~\cite{HanceKrnic2026}. Both models are Kochen--Specker noncontextual \cite{HanceKrnic2026}.) We therefore take
\begin{equation}
R(t)=\frac{\Delta_a(t)}{(1-t)P(a)}
\label{eq:R}
\end{equation}
as the single-shot distance per absorbed photon; $R=1$ for every non-disturbing definite-path model. At $t=0$, $R-1$ is Ref.~\cite{Hance2024}'s counterfactual gain $\Delta_a-P(a)$ divided by $P(a)$. In the weak-absorber limit, the linear response, \cref{eq:linres} gives $|P(m|t)-P(m)|\to(1-t)|\varrho(a,m)|$, and with $\sum_m|\varrho(a,m)|=P(a)+2\Nneg$,
\begin{equation}
\lim_{t\to1}R(t)=1+\frac{\Nneg}{P(a)}, \qquad \Nneg=\sum_{m:\,\varrho(a,m)<0}|\varrho(a,m)|,
\label{eq:Rweak}
\end{equation}
where $\Nneg$ is the negativity of the KD distribution over (path $a$, output). The single-shot gain for a weak absorber equals the negativity per unit absorption, which is unbounded. For a two-path interferometer with input $\sqrt{1-p}\ket{b}+\sqrt p\ket{a}$ and balanced readout, $\Nneg=\sqrt{p}(\sqrt{1-p}-\sqrt p)/2$ for $p\le\tfrac12$, so $R(1^-)=\tfrac12+\tfrac12\sqrt{(1-p)/p}$, which grows without limit as the coupling $p=P(a)$ to the object is weakened. \Cref{fig:disc}(a) shows $R(t)$ for this family at $p=1/20$, for two of the states of the three-path interferometer studied experimentally in Ref.~\cite{Sagawa2026}, and for the balanced Elitzur--Vaidman interferometer. In that device the object sits on the internal path $\ket f=(\ket1+\ket2-\ket3)/\sqrt3$ \cite[Eq.~(1)]{Hofmann2025} and the three outputs $\ket1,\ket2,\ket3$ are measured, so $\varrho(f,m)$ and $P(f)$ replace $\varrho(a,m)$ and $P(a)$ below. The three-path state $\ket{N_f}=(\ket1+\ket2+\ket3)/\sqrt3$ has $\varrho(f,m)=(1,1,-1)/9$ with $P(f)=1/9$, so $R$ falls from $10/3$ (at $t=0$), to $2$ as $t\to1$; the state $\ket{V_0}=(2\ket1+2\ket2+\ket3)/3$ has $\varrho(f,m)=(2,2,-1)/9$ with $P(f)=1/3$, giving $2$ and $4/3$. The Elitzur--Vaidman interferometer has no negative term, $\varrho=(1/2,0)$, and its gain at $t=0$, $R=3/2$, comes entirely from the null port: it disappears as $t\to1$ because the null-port probability is second order in $(1-t)$.

How large can $R$ be at a given $P(a)$? The two hypotheses correspond to two quantum states of the photon and the absorber's record, $\proj\psi$ and $K_t\proj\psi K_t\oplus[1-S(t)]\proj{\mathrm{abs}}$, and the statistical distance of any readout is at most half the trace norm of their difference \cite[Ch.~IV]{Helstrom1976}. The difference has the eigenvalue $-(1-S)$ on the absorbed outcome and, on the span of $\ket\psi$ and $\ket a$, two eigenvalues of opposite sign with sum $1-S$ and product $-(S-\chi^2)$, where $\chi=\bra\psi K_t\ket\psi=1-(1-\sqrt{t})P(a)$. Therefore, for every readout,
\begin{equation}
\Delta_a(t)\le\frac12\Big[(1-S)+\sqrt{(1+S)^2-4\chi^2}\Big],
\label{eq:helstrom}
\end{equation}
with equality for the readout that contains the two eigenvectors, which are real combinations of $\ket\psi$ and $\ket a$ and so define a two-path interferometer. \Cref{eq:helstrom} is the minimum-error probability of Zhou and Yung's channel-discrimination formulation \cite[Eq.~(41)]{Zhou2017}, written for the statistical distance at equal priors and with the overlap of the surviving and object-free states evaluated for a single pass through an arbitrary interferometer. This depends on the interferometer only through $P(a)$, and reaches one, an error-free test, only when $\chi=0$, which is Ref.~\cite[Theorem~3]{Zhou2017}, and for a single pass needs $t=0$ and $P(a)=1$. At $t=0$ it reduces to the bound $\Delta_a-P(a)\le\tfrac12[\sqrt{(4-3P(a))P(a)}-P(a)]$ on the counterfactual gain of an opaque object, Ref.~\cite[Eq.~(14)]{Hance2024}, of which it is the partially transmitting generalisation. There it gives $R\le[P(a)+\sqrt{4P(a)-3P(a)^2}]/[2P(a)]$: $3.37$ at $P(a)=1/9$, against $10/3$ for $\ket{N_f}$, and the golden ratio $(1+\sqrt5)/2$ at $P(a)=1/2$, against $3/2$ for the balanced Elitzur--Vaidman readout. For a weak absorber, $1-S\to2(1-\sqrt{t})P(a)$ and $(1+S)^2-4\chi^2\to4(1-\sqrt{t})^2P(a)$, so
\begin{equation}
\lim_{t\to1}R(t)\le\frac12\Big[1+\frac1{\sqrt{P(a)}}\Big],
\end{equation}
slightly above the value $\tfrac12+\tfrac12\sqrt{(1-p)/p}$ of the balanced readout. With \cref{eq:Rweak} this bounds the negativity itself, for a pure input: a readout cannot have more negativity than a Helstrom measurement has distinguishing power.

\begin{lemma}
\label{lem:negmax}
For every state in every dimension, every path projector $\hat{\Pi}(a)$ and every orthonormal output basis,
\begin{equation}
\Nneg=\sum_{m:\varrho(a,m)<0}|\varrho(a,m)|\le\tfrac12\sqrt{P(a)}\big[1-\sqrt{P(a)}\big]\le\tfrac18 ,
\label{eq:negmax}
\end{equation}
the overall maximum $1/8$ being reached at $P(a)=1/4$. For $P(a)>0$, equality requires $\hat{\rho}$ to be pure.
\end{lemma}

\begin{proof}
$|\KD(a,m)|=|\braket{m}{a}|\,|\bra{a}\hat\rho\ket{m}|\le|\braket{m}{a}|\sqrt{P(a)P(m)}$ by Cauchy--Schwarz in the inner product $\langle x,y\rangle=\bra{x}\hat\rho\ket{y}$, so a second Cauchy--Schwarz with $\sum_m|\braket{m}{a}|^2=1$ and $\sum_mP(m)=1$ gives $\sum_m|\varrho(a,m)|\le\sum_m|\KD(a,m)|\le\sqrt{P(a)}$, while the first sum rule in \cref{eq:sumrules} gives $\sum_m\varrho(a,m)=P(a)$, so $\sum_m|\varrho(a,m)|=P(a)+2\Nneg$ and $2\Nneg\le\sqrt{P(a)}-P(a)$. Equality requires $\hat\rho^{1/2}\ket m\propto\hat\rho^{1/2}\ket a$ at every output with $\braket{m}{a}\neq0$, from the first step, $P(m)=|\braket{m}{a}|^2$ at every $m$, from the second, and every $\KD(a,m)$ real, from $|\mathrm{Re}[z]|\le|z|$. For $P(a)>0$ the first two force $\hat\rho$ to be pure: outputs with $\braket{m}{a}=0$ then have $P(m)=0$, so $\hat\rho\ket m=0$ there, and the range of $\hat\rho^{1/2}$ is spanned by $\hat\rho^{1/2}\ket a\neq0$. For a mixed state the inequality is therefore strict, though the bound remains the supremum.
\end{proof}

\Cref{lem:negmax} bounds the sum of all negative terms over a complete output basis. The same bound on a \emph{single} KD term is already known \cite{Zhang2025,Hofmann2025}, with its minimum over $P(a)$, $-1/8$, in Ref.~\cite{Fernandes2024}, and the first inequality of the proof, $\sum_m|\KD(a,m)|\le\sqrt{P(a)}$, is the term in $a$ of the chain by which Budiyono \textit{et al} prove their Proposition~2 \cite{Budiyono2024}, so \cref{eq:negmax} adds the passage from that $\ell_1$ bound to the negativity. Summed over an orthonormal path basis, \cref{eq:negmax} bounds the total negativity of the KD distribution of path and output by the known value $(\sqrt d-1)/2$ (Appendix~\ref{app:efficiency}). For two paths with $\ket\psi=\cos\vartheta\ket b+\sin\vartheta\ket a$ and outputs $\cos\theta\ket b+\sin\theta\ket a$ and $\sin\theta\ket b-\cos\theta\ket a$, the negative term is $-\tfrac12\sin\vartheta[\sin(2\theta-\vartheta)-\sin\vartheta]$, so \cref{eq:negmax} is saturated at $2\theta=\vartheta+\pi/2$. The three-path interferometer's $\ket{N_f}$ state, with $\Nneg=P(f)=1/9$, saturates \cref{eq:negmax}: no readout of any interferometer with $P(a)=1/9$ has more negativity or a larger single-shot gain for a weak absorber. The state $\ket{V_0}$ has $\Nneg=1/9$ against a bound of $0.122$ at $P(f)=1/3$. At $t=0$ the ranking inverts: $\ket{V_0}$ saturates \cref{eq:helstrom}, with $\Delta_a-P(f)=1/3$ against a bound of $1/3$, while $\ket{N_f}$ gives $7/27$ against $0.264$.

For a weak absorber, the single-shot distance and \cref{sec:fisher}'s Fisher information are built from the same responses $\varrho(a,m)$: \cref{eq:Rweak} sums their moduli, and the first form of \cref{eq:Fs1} sums their squares divided by $P(m)$. Dropping the null-port term of that form and applying the Cauchy--Schwarz inequality over the occupied outputs, whose probabilities sum to one, links the two,
\begin{equation}
\Fs(1)\ge\sum_{m:P(m)>0}\frac{\varrho(a,m)^2}{P(m)}\ge\big[P(a)+2\Nneg\big]^2,
\label{eq:floor}
\end{equation}
since an unoccupied output has $\KD(a,m)=0$, so that the sum of $|\varrho(a,m)|$ over the occupied outputs is $P(a)+2\Nneg$ by the first sum rule. Neither step needs a pure input, and $\Fs(1)\le P(a)$ holds for every input by \cref{eq:mixed}, so \cref{eq:floor} is a third route to \cref{eq:negmax}, beside the $\ell_1$ bound and the statistical distance: negativity is bounded because more of it would force more Fisher information about $t$ out of the unabsorbed photons than direct transmission of the fraction $P(a)$ of the light provides. For $P(a)>0$, \cref{eq:floor}'s equality holds iff $|\varrho(a,m)|/P(m)$ is the same at every occupied output, and there is no null port.

\subsection{The error exponent of a repeated test}
\label{sec:manyshot}

The single-shot distance is the right measure when one photon must decide whether a blocker is present or absent, but it is not additive, meaning we must consider how the error falls when the trial is repeated. For $M$ independent trials, the optimal error probability decays as $e^{-MC}$, where the Chernoff information $C=-\ln\min_{s\in[0,1]}\sum_x P_0(x)^sP_1(x)^{1-s}$ is computed over all outcomes \cite[Sec.~11.9]{Cover2006}. With the object absent the absorbed outcome never occurs, so it contributes nothing to the sum for $s>0$, and outputs that are dark in the object's absence drop out in the same way. Writing $f(s)=\sum_{m:P(m)>0}P(m)^sP(m|t)^{1-s}$, we get $f(1^-)=1$ when the object extinguishes no occupied output, and $f(0^+)=S(t)-\sum_{m:P(m)=0}P(m|t)$, so $C\ge-\ln f(0^+)$. In a non-disturbing definite-path model, $P(m|t)\le P(m)$ for every output, $f$ is increasing, and $C=-\ln S(t)$, meaning the exponent per absorbed photon is $-\ln[1-(1-t)P(a)]/[(1-t)P(a)]$, the same as for direct transmission of a fraction $P(a)$ of the light. Quantum mechanically, $f$ is convex, and $f'(0^+)$ is positive for a weak absorber (Appendix~\ref{app:chernoff}), so $f$ is again increasing and
\begin{equation}
C(t)=-\ln\Big[S(t)-\sum_{m:P(m)=0}P(m|t)\Big],
\label{eq:chernoff}
\end{equation}
which differs from the non-disturbing value only through the null ports, whose contribution is $O((1-t)^2)$. The inequality $C\ge-\ln f(0^+)$ is Ref.~\cite{Wildfeuer2026b}'s lemma for null experiments, with absorption and the null ports as the forbidden outcomes. The equality case stated there as necessary and sufficient, proportionality of the two distributions off the forbidden set, is sufficient but not necessary: the necessary and sufficient condition is $f'(0^+)\ge0$, which \cref{fig:disc}'s three interferometers with negative KD terms meet above the thresholds $t^*$ of panel (b) without proportionality. \Cref{fig:disc}(b) shows the ratio $C(t)/[-\ln S(t)]$ for the four examples, computed numerically (Appendix~\ref{app:chernoff}). Negativity plays no role, meaning, for the weak absorber, the per-absorbed-photon exponent of every scheme, quantum or classical, tends to one. This is \cref{thm:main}'s discrimination counterpart. The single-shot distance is first order in the changes $(1-t)\varrho(a,m)$ and rewards their absolute values, so it counts negative terms twice. The exponent of a repeated test, once the object is weak enough that $f'(0^+)\ge0$, is the absorption probability itself. To second order in $1-t$ the KD terms enter $f'(0^+)$ only through $\sum_m\varrho(a,m)^2/P(m)$, which the second sum rule bounds by $P(a)$, so Appendix~\ref{app:chernoff}'s expansion gives $f'(0^+)\ge\tfrac12(1-t^2)P(a)>0$ for every readout, null ports only adding to it. Negativity therefore cannot switch the gain on perturbatively, and does so only where that expansion fails. The unregistered case is the same. Without null ports the Chernoff information between the Poisson click records is $\tfrac18(1-t)^2\sum_mP(m)(\mathrm{Re}[\wv{m}])^2+O((1-t)^3)$ per photon sent, the same sum again, so direct transmission attains its bound; a null port adds its Elitzur--Vaidman probability at the same order, and the balanced Elitzur--Vaidman interferometer then exceeds direct transmission by the factor $9/8$ per absorbed photon. This does not contradict \cref{cor:absorbed}, which bounds the Fisher information: an output that is empty when the object is absent is worth up to twice the $\tfrac18(1-t)^2F$ that an outcome with Fisher information $F$ is worth. A factor of two is the most any readout can gain in this case, approached when $\ket{\psi}$ is one of the outputs and $P(a)$ is small (Appendix~\ref{app:chernoff}).

For a strongly absorbing object, the expansion fails and negativity does raise the exponent. For $t=0$, the two three-path states and the weakly coupled two-path interferometer exceed the non-disturbing exponent by 12\%, 15\% and 13\%, whereas their single-shot gains in panel (a) are 233\%, 100\% and 386\%. The gain vanishes identically once $f'(0^+)\ge0$, which happens at $t^*=0.102$, $0.106$ and $0.114$ respectively. These thresholds belong to the three interferometers and not to single-pass schemes in general: by \cref{eq:brightcond} an output with small $P(m)$ brightens under a stronger absorber, and it lowers $f'(0^+)$ by about $B_m\ln[B_m/P(m)]$, where $B_m=(1-\sqrt{t})^2|\braket{m}{a}|^2P(a)$ is its Elitzur--Vaidman probability. The threshold therefore rises towards one as an output approaches a null port, but only logarithmically. For the two-path interferometer with $P(a)=1/20$, read out at angle $\delta\theta$ from the one at which a port is dark and on the side where that port's KD term is negative (so that it brightens at every $t$), $t^*=0.14$, $0.28$ and $0.56$ at $\delta\theta=0.3$, $0.03$ and $0.001$, with $1-\sqrt{t^*}\simeq1/\{[1-P(a)]\ln[(1-\sqrt{t^*})\sqrt{P(a)[1-P(a)]}/\delta\theta]\}$ as $\delta\theta\to0$, and $t^*=0.9$ would need $\delta\theta=1.5\times10^{-11}$, a dark-port probability of $2\times10^{-22}$ (Appendix~\ref{app:chernoff}). At any extinction available in practice, the exponent of a weak absorber is the same as in the classical case unless a port is dark by design. The Elitzur--Vaidman interferometer, which has a null port but no negativity, has $f(s)=[(1+\sqrt t)^2/4]^{1-s}$, so $C(t)=-2\ln[(1+\sqrt t)/2]$: twice the non-disturbing value at $t=0$, and an advantage that decays smoothly with $t$ rather than switching off. Wildfeuer found in the Zeno chain that the discrimination advantage requires the competing hypothesis to be the object's absence, since a click in the port that is dark without the object then proves the object present, and that it grows with the number of cycles because the absorbed fraction falls as $\pi^2/(4n)$ while that null stays exact \cite{Wildfeuer2026}. Negativity therefore gives a large single-shot gain that persists only for strong absorbers, and null ports give an accumulating gain that remains, weakly, for any $t<1$.

\subsection{The optimal single-pass readout}
\label{sec:optimalreadout}

\Cref{sec:manyshot} gives the exponent of a fixed readout. How far any readout can push it, with the photon heralded, is given by the quantum Chernoff bound, which also identifies the optimal single-pass readout. The exponent of every measurement on $M$ copies of \cref{sec:singleshot} 's two states is at most, by that bound \cite{Audenaert2007,Nussbaum2009}, $-\ln\min_s\mathrm{Tr}[\hat{\rho}_0^{\,s}\hat{\rho}_1^{1-s}]$. Here $\hat{\rho}_0=\proj{\psi}$ is pure, so the trace is $S^{-s}\chi^2$ with $\chi=1-(1-\sqrt{t})P(a)$ as before, smallest as $s\to0$, and
\begin{equation}
C(t)\le-2\ln\big[1-(1-\sqrt{t})P(a)\big]
\label{eq:qcb}
\end{equation}
for every single-pass readout. The bound is reached when $\ket{\psi}$ is itself one of the outputs. Every other output is then a null port, the null ports together have probability $S-\chi^2=(1-\sqrt{t})^2P(a)[1-P(a)]$, the object's absence is supported on a single outcome, and $f(s)=\chi^{2(1-s)}$. This is the equality case of Wildfeuer's lemma for null experiments \cite{Wildfeuer2026b}. At $P(a)=1/2$, \cref{eq:qcb} is the Elitzur--Vaidman exponent $-2\ln[(1+\sqrt t)/2]$ found above, so the balanced Elitzur--Vaidman interferometer is optimal among all readouts with its $P(a)$. Since $-\ln(1-x)$ is convex, \cref{eq:qcb} is at most $2/(1+\sqrt t)$ times the non-disturbing exponent $-\ln S$, reaching equality as $P(a)\to0$, implying a factor of two for an opaque object, whatever $P(a)$, and no gain for a weak absorber. The gains of $12\%$ to $15\%$ from negativity at $t=0$ should be compared with this factor of two, given the readout that reaches it has no negative KD term. For many photons, the resource is the null port, with \cref{eq:qcb} quantifying the information it can deliver in a single pass.

The null-port gain needs absence as the competing hypothesis. Between transmittances $t$ and $t-\delta$ with $\delta$ small, the quantum Chernoff exponent of the two heralded states is $\delta^2[F_{\mathrm{flux}}/8+S\mathcal Q/4]$ to leading order, where $F_{\mathrm{flux}}=P(a)^2/[S(1-S)]$ is the Fisher information of the record of survival or absorption and $S\mathcal Q=P(a)[1-P(a)]/(tS)$ that of the surviving state. Their sum $F_{\mathrm{flux}}+S\mathcal Q=P(a)/[t(1-t)]$ is \cref{sec:bound}'s registered bound. A readout limited by its Fisher information $F$ has the exponent $\delta^2F/8$, so the most any measurement gains over direct transmission is the factor $1+[1-P(a)](1-t)/S$. The pure surviving state is worth twice its Fisher information, as a null port tuned to $t$ would show, but for a weak absorber it holds a fraction of order $1-t$ of the total. This is the single-pass form of Ref.~\cite[Fig.~4]{Wildfeuer2026}'s flat family, the case whose absorbed-photon count Mitchison, Massar and Pironio bounded for arbitrary protocols \cite{Mitchison2002} and whose close-transparency form is Ref.~\cite{Massar2001}'s bound.

Appendices~\ref{app:finiteM} and~\ref{app:sequential} derive how long the single-shot gain lasts, and what a sequential test with an adaptively chosen interferometer can add. For a weak absorber, the gain over $M$ heralded trials is gone once three or four photons' worth of probability has been sent down the object path, and never lowers the error probability by more than the fraction $(1-t)/4$. In the unregistered case, an adaptive choice of interferometer gains nothing beyond the best fixed choice, negativity redistributes the evidence per absorbed photon as it does the Fisher information, and a gain of order one needs an imperfect null port. Multi-pass and multi-photon strategies, which do help with the loss unregistered (\cref{sec:multipass}), are beyond this analysis, and so is the registered case, where the first absorbed photon is itself a proof (Appendix~\ref{app:sequential}). \Cref{tab:map} summarises \cref{sec:fisher,sec:helps,sec:discrimination}'s bounds.

\section{Testing both corollaries on one interferometer}
\label{sec:test}

\cref{sec:fisher,sec:discrimination}'s results apply to \cref{sec:singleshot} and Ref.~\cite{Sagawa2026}'s silicon three-path interferometer, where the transmittance of path $\ket{f}$ is set by a Mach--Zehnder stage, and each of the three outputs is detected. \Cref{cor:detected} can be tested on the published device, and \cref{cor:absorbed} could be tested on a higher-visibility version of it. Recording the output rates over the full range of $t$, and fitting them against $\epsilon=1-\sqrt t$ (in which \cref{eq:Pmt} is quadratic), returns $P(m)$, $\varrho(f,m)$ and $|\braket{m}{f}|^2P(f)$ port by port, and with them $\partial_t\ln P(m|t)|_{t=1}$, which \cref{eq:linres} predicts to be $\mathrm{Re}[\wv{m}]$: $(1,1,-1)/3$ for $\ket{N_f}$ and $(1/2,1/2,-1)$ for $\ket{V_0}$. The negative entry determines the port that brightens, its squared slope giving \cref{cor:detected}'s information per detected photon, and the rate-weighted sum of the squared slopes should equal $P(f)$, which is \cref{eq:Fs1} and shown by \cref{cor:absorbed}. Since $|\KD(f,m)|^2=P(m)|\braket{m}{f}|^2P(f)$, that equality holds port by port, and in the sum, iff the KD terms are real, so this test measures \cref{thm:main}'s equality condition. From the Fisher matrix of that fit, with ten settings spanning $t\in[0,1]$ and the blocked and open ends included, establishing the negative slope at five standard deviations needs $2\times10^4$ heralded photons for $\ket{V_0}$ and $3\times10^4$ for $\ket{N_f}$, and confirming the sum over ports to $2\%$ needs $1\times10^6$ and $2\times10^6$. At the $100\,\mathrm{s}$ per setting used in Ref.~\cite{Sagawa2026}, the last requires a heralded rate of $2\times10^3\,\mathrm{s^{-1}}$. The same data gives \cref{sec:singleshot}'s single-shot distance, whose value at $t=0$ can already be computed from Ref.~\cite{Sagawa2026}. Their output distributions with path $\ket f$ open and blocked give $R(0)=3.4$ for $\ket{N_f}$ and $2.0$ for $\ket{V_0}$, against the ideal $10/3$ and $2$. For their third input, $\ket{B_f}=(2\ket1+2\ket2+3\ket3)/\sqrt{17}$, with $P(f)=1/51$ and an ideal $R(0)=22/3$, they give $3.3$, because the absorbed fraction inferred from the blocked distribution is $0.053$ rather than $1/51$. Since $R$ divides by the absorbed fraction, the largest single-shot gains, at small $P(f)$, are also the most sensitive to errors in it.

Testing for an anomalous weak value, a negative $\mathrm{Re}[\wv{m}]$, needs only its sign and is limited by counting statistics; neither state has $|\mathrm{Re}[\wv{m}]|>1$, so \cref{cor:detected}'s amplification, more than one unit per detected photon, would need a darker port. \cref{cor:absorbed}'s sum is an equality which is limited instead by systematic errors: reading $P(f)$ off the summed quadratic coefficients of the same fit cancels the port efficiencies and any multiplicative error in the $t$ scale, but decoherence depresses the sum by an amount of order $1-V$ for a visibility $V$, which is $\tfrac{10}3(1-V)$ for $\ket{N_f}$ and $\tfrac43(1-V)$ for $\ket{V_0}$ under dephasing in the basis $\{\ket1,\ket2,\ket3\}$, so the $2\%$ target asks for $V$ above $0.994$ and $0.985$ respectively. Dephasing between $\ket f$ and its complement, the form of the visibilities fitted in Ref.~\cite{Sagawa2026}, gives $\tfrac{16}9(1-V)$ and $3(1-V)$ instead, and asks for $V$ above $0.989$ and $0.993$. The visibilities fitted port by port in Ref.~\cite{Sagawa2026}, $0.92$ to $1.21$, are mostly several times further from one than either, so this half of the test needs a better device than the published one. The sign test is far more tolerant: the negative entry persists down to $V=1/2$ for $\ket{N_f}$ and $V=1/4$ for $\ket{V_0}$ under the first form of dephasing, and down to $V=1/4$ and $V=1/2$ under the second. A nonlinear $t$ scale is self-diagnosing, since \cref{eq:Pmt} is quadratic in $\epsilon$ and a distorted scale pushes some port-by-port ratios above one, which no state allows. \Cref{thm:main} fails if the sum exceeds $P(f)$, or if a shortfall larger than the measured visibility allows persists at every phase on path $\ket f$; the phase scan diagnoses an imaginary $\KD(f,m)$ that comes from an uncompensated phase in that arm, and a preparation complex in more than one relative phase is not diagnosed this way.

\begin{table*}[t!]
\caption{Information per absorbed photon for an object of transmittance $t$ in path $a$, with the absorptions unregistered or registered. ``Unregistered'' means a source of known mean whose absorbed photons are not detected as such (Poisson counts at the outputs); ``registered'' means the photon number is known, so an absorption is itself an outcome. Single-pass entries are results of this paper except where a reference is given; the last column gives what is known for many coherent passes of one photon. A dash means we don't consider that metric for that scenario.}
\label{tab:map}
\begin{ruledtabular}
\begin{tabular}{>{\raggedright\arraybackslash}p{0.21\textwidth}| >{\raggedright\arraybackslash}p{0.1\textwidth}| >{\raggedright\arraybackslash}p{0.34\textwidth}| >{\raggedright\arraybackslash}p{0.29\textwidth}}
Figure of merit & Absorptions & Single pass, any interferometer & Many passes of one photon \\
\hline \hline
Fisher information per absorbed photon (\cref{cor:absorbed}) & unregistered & $\le 1/[t(1-t)]$; reaching equality for real KD terms and pure inputs, \cref{eq:mixed} otherwise, though a path-basis readout attains it for every input & below the registered bound at any finite number of passes, and no better per absorbed photon than direct transmission unless the responses of different passes interfere (\cref{prop:manypass}); at most $0.932$ of that bound in the Zeno model, at $n\simeq2.01/(1-\sqrt t)$ cycles (Appendix~\ref{app:zeno}, cf.\ Ref.~\cite{Thomas2014}), and $0.648$ for repeated transmission as $t\to1$ \cite{Wildfeuer2026,Wildfeuer2026b}; the registered bound itself for a ring resonator at critical coupling on resonance, in the limit of a long wavepacket (\cref{sec:multipass}), with time-resolved detection, no background, a narrow source and no loss but the object's \\
\hline
Fisher information about $c$ per absorbed photon, dispersive object (\cref{eq:disp-abs}) & unregistered & $\le t'^2/[t(1-t)]$ $+\,4t\phi'^2[1-P(a)]/[(1-t)S]$; attained for a pure input, by a readout that typically depends on $c$; the phase term needs an interferometer, and gains up to $1.54\,x_0$ over the best direct transmission for a Lorentzian sample of resonant optical depth $x_0\ge1.59$ & phase accumulates over passes \cite{Wildfeuer2026b,Yu2026} \\
\hline
Fisher information per absorbed photon & registered & $\le 1/[t(1-t)^2]$ \cite{Massar2001,Facchi2002,Wildfeuer2026}; for one photon, equality for real KD terms and pure inputs (this paper); the bound itself holds for any multi-photon input \cite{Nair2018} & same bound for any fixed arrangement \cite{Wildfeuer2026} \\
\hline
Single-shot distance per absorbed photon (\cref{eq:Rweak,eq:helstrom}) & registered & $1+\Nneg/P(a)\le\tfrac12[1+1/\sqrt{P(a)}]$ as $t\to1$; gone after $M\simeq3/P(a)$ trials in \cref{eq:RMgauss}'s Gaussian regime (Appendix~\ref{app:finiteM}) & -- \\
\hline
Error exponent per absorbed photon (\cref{eq:chernoff,eq:qcb}) & registered & classical value for $t$ above an interferometer-dependent threshold $t^*$ ($\simeq0.10$--$0.11$ for \cref{fig:disc}'s three interferometers with negative KD terms, rising towards one, but only logarithmically, as a port darkens), negativity raising it for $t<t^*$; null ports add to it at every $t<1$, and no readout exceeds the classical value by more than the factor $2/(1+\sqrt t)$, approached with $\ket\psi$ as an output as $P(a)\to0$ & grows as $n\ln n$ in the number $n$ of cycles against empty space, and not at all between two partial transparencies \cite{Wildfeuer2026} \\
\hline
Error exponent per absorbed photon (Appendix~\ref{app:chernoff}) & unregistered & at most the direct-transmission value without null ports, to leading order in $1-t$, with equality for real KD terms and pure inputs; up to twice that with them & -- \\
\hline
Sequential evidence per absorbed photon (\cref{eq:etaweak}) & unregistered & without null ports, at most the direct value to first order in $1-t$; a gain of order one needs $(1-t)|\wv{m}|\gtrsim1$, then logarithmic in the extinction ratio; a null port gives proof after $(1+\sqrt t)/[(1-\sqrt t)|\braket{m}{a}|^2]$ absorbed photons on average & absorbed photons before proof $\to0$ as the number of passes grows \cite{Mitchison2001} \\
\end{tabular}
\end{ruledtabular}
\end{table*}

\section{Discussion}
\label{sec:discussion}

Is KD negativity the resource behind the gentleness (higher-information-per-absorbed-photon) of interaction-free measurement? 

For estimating how strongly a sample absorbs, it is not. The KD distribution of path and output is the outputs' linear response to a weak sample (\cref{eq:response}), which \cref{thm:main} turns into the Fisher information. Negativity cancels from the total information, instead only moving information between outputs, concentrating it, per detected photon, in weakly lit ones. The imaginary parts of the KD distribution, which are defined by the relation between the absorber path and the readout, instead reduce the total information. Lost coherence in the input reduces information further, and several absorbers are estimated no better than by direct transmission of each one's share of the light (\cref{eq:fim}). The imaginary part is also the imaging photon's response to a phase, which is the only way a single pass gains per absorbed photon, for an optically thick or inhomogeneously broadened sample (\cref{thm:dispersive}). Many passes of one photon gain only through the interference of the responses of different passes, and only with the absorptions unregistered (\cref{prop:manypass}). If the absorptions are instead registered, Ref.~\cite{Wildfeuer2026} rules out a gain for any fixed arrangement. 

For detection, negativity gives a single-shot gain per absorbed photon equal, for a weak absorber, to the negativity per unit absorption, which is at most $\tfrac12[1/\sqrt{P(a)}-1]$. That gain disappears after about $3/P(a)$ photons when $P(a)$ is small and no bright output is occupied more weakly than the object path, and for a readout without null ports, ideal or imperfect, it does not accumulate beyond a correction of relative order $1-t$. Any information gain which accumulates comes from null ports, ideal or imperfect, including when the test is sequential and the interferometer is chosen adaptively (Appendix~\ref{app:sequential}). The optimal single-pass test, with $\ket{\psi}$ as one output and null ports for the rest, has no negative KD term. \Cref{sec:test} shows that the identity behind \cref{cor:detected} can be tested on an existing silicon three-path interferometer and \cref{cor:absorbed} on a higher-visibility version of it.

Counterfactual gain and KD negativity are resources when the object we want to detect is strongly absorbing and the question is whether it is there, when a single photon must decide, or when photons at the detector rather than photons in the sample are the scarce quantity \cite{LupuGladstein2022,Fukuda2026}. In the last case, \cref{cor:detected} tells us to look at the output with the largest $|\mathrm{Re}[\wv{m}]|$ for the blocked-path projector, which the sensor's own linear response identifies. That output beats direct transmission per detected photon, for a weak absorber, only when the value exceeds one, in which case the weak value is anomalous. For estimation, this links our result to Ref.~\cite{ArvidssonShukur2020}: there, negativity of the doubly extended distribution is necessary for a postselected gain per retained particle for a unitary parameter; here, an anomalous weak value, which for a real readout is a negative entry of that doubly extended distribution for the path projector, is necessary for a gain per detected photon, but negativity gives no gain per absorbed photon.

A phase or a polarisation rotation is imprinted unitarily and accumulates coherently over passes, which is why many passes help for a phase \cite{Wildfeuer2026,Wildfeuer2026b,Yu2026}, why the Zeno chain responds to a phase shift \cite{Thomas2014}, and why phases help for absorption itself when the losses go unregistered (\cref{sec:multipass}). Interaction-free polarimetry of a birefringent but transparent sample is therefore not covered by the absorption bounds here, and remains a case where a counterfactual protocol can reduce the light deposited in the sample per unit of information. For a single pass, \cref{eq:disp-abs} show that the phase term is largest when little of the light meets the sample, where for a Lorentzian line its ratio to the absorption term at the same detuning is the square of the detuning in half-widths. Against direct transmission of a weak sample on resonance, there is no gain per absorbed photon for a Lorentzian line, and, with the absorptions unregistered, about a factor of $50$ for a room-temperature caesium cell probed in the wing. Against the best direct transmission of a Lorentzian sample of resonant optical depth $x_0$, the gain is up to $1.54\,x_0$ for $x_0\ge1.59$ with the absorptions unregistered, and up to $1+x_0$ with them registered (see Appendix~\ref{app:dispersive} for more details).

For a sensor designer aiming to make a low-dose imaging protocol, the results of our analysis come down to a few rules. To estimate an absorption with the least dose in one pass, send the light straight through, or use an interferometer whose KD terms are real, fed with a pure input: interference then adds nothing, and complex KD terms or lost coherence subtract information from the total this provides (\cref{eq:mixed,eq:visibility}). Many passes improve the device only when the absorbed photons are not registered, and then a critically coupled ring on resonance comes close to the registered bound for a long wavepacket, given a prior estimate of $t$ and a coupling that tolerates the detuning, and given that little light reaches its dark port from anything but the object: a background, a source broader than a small fraction of the ring's linewidth, and loss in the ring other than the sample's each lower the gain, and counting the whole wavepacket without time resolution removes it at critical coupling, though moving the coupling off critical recovers most of it (\cref{sec:multipass}, Appendix~\ref{app:zeno}). To gain in one pass the quantity sought must also shift the phase, meaning for an optical sample, whose detuning can be chosen, the sample must be optically thick or inhomogeneously broadened (\cref{sec:dispersive}). Electron microscopy of biological specimens is a case where the quantity sought shifts the phase and the damage comes from inelastic events that outnumber the elastic ones forming the image \cite{Henderson1995}; interaction-free designs using many passes have been proposed for opaque and transparent pixels \cite{PutnamYanik2009,Kruit2016}, with the specimen's phase shift noted but not analysed \cite{Kruit2016}, and interaction-free measurement with electrons has been demonstrated \cite{Turner2021}. If the inelastically scattered electrons are removed from the recorded beam, \cref{thm:dispersive} bounds any single-pass design per inelastic event and \cref{prop:manypass} says where a multi-pass design can gain, but evaluating either needs a model of how transmittance and phase depend on the specimen, which we have not attempted.

Everything above is counted per absorbed photon, but two of the results quoted in the Introduction are counted using a different metric: twin-beam probes gain information per photon incident on the sample, and Ref.~\cite{Belsley2022}'s ring resonator gains information per photon entering the device. In a single pass, the object absorbs the fraction $1-t$ of the light incident on it, whatever the state, so the twin-beam device improvement over classical light, up to the factor $1/(1-t)$, is the same per absorbed photon: the idler registers the photon number, and with ideal detection the twin beam reaches the registered bound, as a heralded single photon does. In the ring, counting per absorbed photon turns the advantage over every single-pass quantum probe into at best a tie with the heralded single photon, reached on resonance at critical coupling for a long wavepacket. It however doesn't affect the gain over classical light passed once, which comes from the interference between passes: with unregistered loss the ring then closes the factor $1/(1-t)$ that the twin beam closes by registering (\cref{sec:multipass}). Neither exceeds the registered bound, nor can any parallel probe \cite{Nair2018}, any fixed arrangement of one photon \cite{Wildfeuer2026}, or any coherent light. Per photon sent is the right normalisation to use when the source provides few high-quality photons at the right wavelength, or when the measurement is limited by the rate a detector can take. Per absorbed photon is the right normalisation to use when the sample is delicate, since a photon that goes through unabsorbed does nothing to the sample, so the damage a measurement does to the sample is set by its absorption.

This work leaves open three questions. The first is the status of the null port. \Cref{eq:brightcond} places it at the end of a family that begins with negative KD terms, and \cref{eq:qcb} makes it the optimal readout, but its own KD term vanishes. Negativity at the null port comes from the sum of the interference terms between path $a$ and the other paths, $\sum_{b\neq a}\mathrm{Re}[\braket{m}{a}\bra{a}\hat\rho\ket{b}\braket{b}{m}]$, which cancels the path's own contribution $|\braket{m}{a}|^2P(a)$ to the KD term. A click in a null port is not a proof of contextuality in the same sense that an anomalous weak value is \cite{Pusey2014}: the pre- and postselection are orthogonal there, whereas that proof assumes they are non-orthogonal. For an opaque object the Elitzur--Vaidman statistics, null port included, are reproduced by models in which the object randomises a hidden variable as well as removing the photon \cite{Catani2023,HanceKrnic2026}, and these models are Kochen--Specker noncontextual \cite{HanceKrnic2026} (\cref{sec:singleshot}). The null port instead violates the weaker non-disturbance assumption behind \cref{eq:R}'s baseline, and where that leaves it relative to \cref{eq:brightcond}'s negative KD terms, of which it is the limit, we have not settled.

The second is what happens when we make many passes with several photons. For a fixed arrangement of one photon, \cref{prop:manypass} puts the whole of the unregistered gain in the interference between passes, which means to find the optimal set-up we just maximise the norm in \cref{eq:manypass} against the absorption below it, over the intermediate optics and the number of passes. The supremum is the registered bound itself, approached by \cref{sec:multipass}'s critically coupled ring on resonance in the limit of a long wavepacket, so the chain only reaches $G^*=0.932$. The best arrangement at a fixed finite number of passes, where every arrangement falls short, is not known. Choosing the optics adaptively does not affect this: deferring the intermediate measurements onto a record register restores \cref{eq:manypass} with $\proj{a}$ replaced by $\proj{a}\otimes\mathds{1}$ and the record read as one more output, though the orthogonality that removes the gain has then to be tested on that larger space. An input of several photons escapes this only when its photons are correlated. A coherent state however does not: through any passive arrangement it leaves every output in a coherent state of amplitude $\sqrt N z_m(t)$, so the Fisher information of its output counts and the photons the object absorbs from it both scale with the mean photon number $N$, and every bound here holds per absorbed photon unchanged. A squeezed or path-entangled state through a many-pass arrangement is outside everything here, which is the case of most practical interest, since a dose-limited measurement uses many passes and does not register what the sample absorbs.

The third is what happens if, instead of trying to determine if a blockage is present or absent, we are trying to localise a blockage to one of a set of candidate locations. \cref{fig:setup} classifies outputs by the sign of the KD term for one blocked path, and \cref{thm:main} classifies readouts by the KD term's reality. For the multi-hypothesis problem of locating an unknown blockage among many paths \cite{Underwood2026,Chaturvedi2026}, the relevant object is instead the KD matrix over all paths. By \cref{eq:negtot}, its total negativity is at most $(\sqrt d-1)/2$, but how that negativity and the matrix's null structure trade off there is not yet known. A second coherent pass there raises the localisation probability for an opaque object above the one-shot optimum with no intermediate measurement needed \cite{Chaturvedi2026}, using a figure of merit not normalised per absorbed photon.

Beyond these questions, there are three generalisations which we have not considered: partial registration of the absorptions, which interpolates between the two cases; joint estimation of loss and phase as two parameters, where \cref{eq:Fsc} becomes a matrix; and a dispersive sample inside a ring, where the buildup and the dispersive term of \cref{eq:disp-abs} combine.

\textit{Acknowledgements---} The author acknowledges support from a Royal Society Research Grant (RG/R1/251590), an EPSRC Mathematical Sciences Small Grant (UKRI3647), and their EPSRC Quantum Technologies Career Acceleration Fellowship (UKRI1217). This work was carried out in cooperation with SPRIN-D, the Federal Agency for Breakthrough Innovation, within the SPRIN-D Funke QuantumSense Exploration.

\textit{Data and Code Availability---} A Wolfram Mathematica notebook reproducing all computations in this paper is openly available on Zenodo~\cite{Hance2026NegImagImagingCode}.

\bibliographystyle{unsrturl}
\bibliography{ref.bib}

\appendix

\section{Mixed inputs: how a background reduces each port's share}
\label{app:mixed}

\Cref{eq:mixed} says that each output keeps at most the fraction $R_m=|\bra{a}\hat\rho K_t\ket{m}|^2/[P(a)P(m|t)]\le1$ of the share $|\braket{m}{a}|^2P(a)/t$ it holds, and keeps that fraction when its KD term is real; an imaginary part takes a further $[\mathrm{Im}[\KD(a,m)]]^2/[t\,P(m|t)]$ from the port. Let $\ket{\psi_a}\propto\hat\rho\ket{a}$ be the part of the input that path $a$ is coherent with, and $\lambda_a=\|\hat\rho\ket{a}\|^2/P(a)$ its weight, so that $\hat\rho=\lambda_a\proj{\psi_a}+(1-\lambda_a)\hat\sigma$ with $\hat\sigma\ket{a}=0$; the remainder is a state because $\hat\rho\ge\hat\rho\proj{a}\hat\rho/P(a)$. Since $\hat\sigma\ket{a}=0$ implies $K_t\hat\sigma K_t=\hat\sigma$,
\begin{equation}
R_m=\frac{\lambda_a\,|\bra{m}K_t\ket{\psi_a}|^2}{P(m|t)}=1-\frac{(1-\lambda_a)\bra{m}\hat\sigma\ket{m}}{P(m|t)},
\label{eq:share}
\end{equation}
the fraction of output $m$'s counts that come from that part, so that for real KD terms
\begin{equation}
\Fs(t)=\frac{P(a)}{t}\Big[1-\hspace{-6pt}\sum_{m:P(m|t)>0}\frac{|\braket{m}{a}|^2(1-\lambda_a)\bra{m}\hat\sigma\ket{m}}{P(m|t)}\Big]:
\label{eq:docked}
\end{equation}
each port's share of path $a$ is reduced by the fraction of its counts the background supplies. In particular, \cref{eq:share} gives $R_m=1$ for a pure input, where $\lambda_a=1$, and, by continuity, at every output with $P(m|t)=0$, which has no background either. An output that $\ket{\psi_a}$ leaves dark receives nothing from it but the Elitzur--Vaidman probability $(1-\sqrt{t})^2|\braket{m}{a}|^2P(a)$, so it forfeits its share once the background exceeds that. A port that is bright but only faintly, or one at a transmittance where $\ket{\psi_a}$ darkens it, is affected in the same way by increased background. By \cref{eq:docked}, the pure-input value is recovered only where the background is dark in every port that path $a$ reaches. In \cref{eq:visibility}'s balanced interferometer, $\ket{\psi_a}\propto\ket{a}+V\ket{b}$ and $\lambda_a=(1+V^2)/2$, so at $t\to1$ the port that is dark at $V=1$ keeps $(1-V)/2$ of the half of path $a$ it holds and the bright port keeps $(1+V)/2$ of its half: for any $V<1$ the two add to one half, which is why half the information is lost as $t\to1$ (\cref{sec:extensions}). Equality now requires, besides real KD terms, that $\hat\rho^{1/2}K_t\ket{m}$ be parallel to $\hat\rho^{1/2}\ket{a}$ at every output with $\braket{m}{a}\neq0$, which for an input of full rank holds only when $\ket{a}$ is itself an output. An ideal null port doesn't lose the information: $P(m)=0$, that is a port dark at $t=1$, forces $\hat\rho\ket{m}=0$, and its bracket in \cref{eq:mixed} vanishes at every $t<1$. What instead loses the information is a port that the coherence would leave dark but that a background fills instead.

\section{Several objects: the restriction on \texorpdfstring{$\mathbf t$}{t} and the row sum rule}
\label{app:several}

\cref{eq:fim}'s restriction to transmittances at which every output is occupied is needed: unlike the single-parameter case, an output with $z_m(\mathbf t)=0$ cannot be assigned the limit of its term, because that limit depends on the direction of approach. In \cref{eq:amp}, $z_m(t)$ moves along a line, so it can vanish only if $\braket{m}{\psi}$ is a real multiple of $\braket{m}{a}\braket{a}{\psi}$, and then $\mathrm{Im}[\KD(a,m)]=0$ at every $t$ and the subtracted term vanishes identically; with several objects $z_m$ can vanish with the $\braket{m}{j}\braket{j}{\psi}$ in general position. The matrix inequality itself is unaffected, the subtracted matrix being positive semidefinite along every direction.

Let $\KD(j,m)=\braket{m}{j}\bra{j}\hat\rho\ket{m}$ be the KD term of path $j$ and output $m$, and $\varrho(j,m)=\mathrm{Re}[\KD(j,m)]$ its real part. With all objects weak, $\mathrm{Im}[\bar z_m\partial_{t_j}z_m]=\tfrac12\mathrm{Im}[\KD(j,m)]$ and $F_{jk}=\sum_m\varrho(j,m)\varrho(k,m)/P(m)$ when there are no null ports;\cref{eq:fim} then comes from the sum rule
\begin{equation}
\sum_{m:P(m)>0}\frac{\KD(j,m)\overline{\KD(k,m)}}{P(m)}=\bra{j}\hat\rho\ket{k}\,\bra{k}\hat\Pi_+\ket{j},
\label{eq:roworth}
\end{equation}
for which \cref{eq:sumrules}'s second rule is the case $j=k=a$. Where no output is dark, $\hat\Pi_+=\mathds1$ in \cref{eq:roworth} and the rows of the KD matrix, weighted by $1/\sqrt{P(m)}$, are orthogonal, whatever their signs; each dark output subtracts $\braket{k}{m}\braket{m}{j}$ from $\bra{k}\hat\Pi_+\ket{j}$. \Cref{eq:roworth} uses the pure input: for a mixed one its left-hand side does not factorise, and the identity fails already for $\hat\rho=\mathds1/d$.

\section{Detector efficiency, multi-photon inputs, the decomposition behind Corollary~\ref{cor:detected} and known bounds related to Lemma~\ref{lem:negmax}}
\label{app:efficiency}

\textit{Detector efficiency.} Detectors of efficiency $\zeta$, the same at every output and for direct transmission, change neither conclusion of \cref{sec:bound}: direct transmission stays best in both cases. Without registration every term of \cref{eq:Fsdef} is multiplied by $\zeta$, for every readout and for direct transmission alike. With registration a missed photon joins the absorbed ones in the no-click outcome, of probability $1-\zeta S$, and the total per absorbed photon, $[\zeta\Fs+\zeta^2P(a)^2/(1-\zeta S)]/[(1-t)P(a)]$, lies below the direct-transmission value $\zeta/[t(1-t)(1-\zeta t)]$ by at least $\zeta^2(1-\zeta)[1-P(a)]/[(1-t)(1-\zeta t)(1-\zeta S)]$.

\textit{Multi-photon inputs.} Nair's bound constrains only the mean photon number in the modes that meet the loss and allows lossless ancilla modes entangled with them; the other paths of a single-pass interferometer are such ancillas, so the registered bound holds for every multi-photon input to one, entangled across the paths or not, and not only for the single photons and coherent light considered here \cite{Nair2018}. The weak-value form \cref{eq:LR} of the likelihood ratio is more restricted: it holds click by click for Fock states of the mode $\ket\psi$ and for coherent light, but not for path-entangled inputs.

\textit{The conditional-uncertainty decomposition.} Eq.~(4) of Ref.~\cite{Ji2026}, with Eqs.~(5) and (6), is the decomposition of the variance [Eq.~(30) of Ref.~\cite{Hall2004}] into the average Ozawa--Hall conditional uncertainty \cite{Ozawa2003,Hall2004} and the variance of the conditional averages: for a rank-one generator and a complete rank-one readout without null ports their Eq.~(5) gives a conditional uncertainty $(\mathrm{Im}[\wv{m}])^2$ at every output, and the decomposition is \cref{eq:Fs1} with $P(a)^2$, the squared mean of $\hat\Pi(a)$, subtracted from both sides.

\textit{Known bounds related to Lemma~\ref{lem:negmax}.} Each KD term is a third-order Bargmann invariant, whose values at fixed $|\braket{a}{\psi}|$ fill an elliptical disc \cite{Zhang2025} with leftmost point $-\tfrac12\sqrt{P(a)}[1-\sqrt{P(a)}]$, the same function of $P(a)$ as \cref{eq:negmax}. Hofmann gives the single-term bound, for real states, in the equivalent form $|\braket{i}{j}|(1-|\braket{i}{j}|)/2$ at fixed overlap of the two projectors \cite[Eq.~(22)]{Hofmann2025}, and minimising over that overlap makes $-1/8$ the most negative entry any KD distribution can have \cite{Fernandes2024}. For a qubit, the negativity summed over both rows of the Margenau--Hill distribution is likewise at most $1/8$ \cite[Prop.~3]{Chakrabarty2026}. Proposition~2 of Ref.~\cite{Budiyono2024} bounds the KD nonclassicality $\sum_{a,m}|\KD(a,m)|-1$, maximised over the output basis, by $\sum_a\sqrt{P(a)}-1$, and holds for every path basis and for mixed states. Applied to the term in $a$ of its proof, $\sum_m|\KD(a,m)|\le\sqrt{P(a)}$, the first sum rule in \cref{eq:sumrules} and $|\mathrm{Re}[z]|\le|z|$ give $P(a)+2\Nneg\le\sqrt{P(a)}$, turning a bound on the modulus sum, which conflates nonreality with negativity, into the bound on $\Nneg$ that \cref{eq:Rweak} needs. Summing \cref{eq:negmax} over an orthonormal path basis, with $\Nneg(a)$ the negativity of path $a$ and $\Nneg_{\mathrm{tot}}=\sum_a\Nneg(a)$, and applying Cauchy--Schwarz to $\sum_a\sqrt{P(a)}$ gives
\begin{equation}
\Nneg_{\mathrm{tot}}\le\tfrac12\Big[\sum_a\sqrt{P(a)}-1\Big]\le\frac{\sqrt d-1}{2},
\label{eq:negtot}
\end{equation}
with equality in the last step if and only if every $P(a)=1/d$. Since $\sum_{a,m}|\KD(a,m)|\ge1+2\Nneg_{\mathrm{tot}}$, both bounds also follow from Proposition~2 of Ref.~\cite{Budiyono2024}, and the second from Theorem~2 of Ref.~\cite{ArvidssonShukur2021}, whose authors show it is attained when three real mutually unbiased bases exist, as at $d=4$.

\section{Dispersive samples}
\label{app:dispersive}

This appendix gives the numbers behind the dispersive samples of \cref{sec:dispersive}: a Lorentzian line of arbitrary optical depth, molecular bands and ring resonators, and a caesium vapour cell.

\textit{Optically thick samples.} With $x_0=2c\kappa/\gamma$ the resonant optical depth and $x=x_0\gamma^2/(\gamma^2+\delta^2)$ the optical depth at detuning $\delta$, the two terms of \cref{eq:disp-abs} sum, as $P(a)\to0$, to $(x_0/c^2)\,x/(e^x-1)$, which rises far from resonance to $x_0/c^2$, the value resonant direct transmission gives only for a weak sample. At finite $P(a)$ the factor $[1-P(a)]/S$ in \cref{eq:disp-abs} lowers this far-detuned value to $[1-P(a)]\,x_0/c^2$. With the absorptions unregistered, direct transmission gives only the first term, $x^2/[c^2(e^x-1)]$, which is $1/c^2$ times the function $y^2/(e^y-1)$ of repeated transmission \cite{Wildfeuer2026b} (Appendix~\ref{app:zeno}, \cref{fig:zeno}), since a thick sample is many thin ones in series; its maximum over the detuning is $0.648/c^2$, at $x=1.59$, when $x_0\ge1.59$, and $x_0^2/[c^2(e^{x_0}-1)]$, on resonance, when $x_0<1.59$. Against the best direct transmission, a single pass with a phase readout therefore gains up to $x_0/0.648=1.54\,x_0$ per absorbed photon when $x_0\ge1.59$, and up to $(e^{x_0}-1)/x_0$ for a thinner sample, which tends to one, the weak-sample result of \cref{sec:dispersive}. With the absorptions registered, direct transmission is best far from resonance, where its term $x^2e^x/[c^2(e^x-1)^2]$ tends to $1/c^2$, and there the phase readout gains $1+[1-P(a)]\,x_0$ over it. The registered first term and the phase term, $x^2e^x/(e^x-1)^2$ and $x(x_0-x)/(e^x-1)$ in units of $1/c^2$, both fall as $x$ rises, and the phase term's factor $[1-P(a)]/S$ is at most one, so no detuning gains more than $1+x_0$, the far-detuned limit as $P(a)\to0$; for $P(a)>1/2+1/x_0$ an intermediate detuning beats the far-detuned value. The far-detuned gain is approached only where the sample is thin at the detuning used, about $90\%$ of it at $x=0.2$, and there the absorption per trial, $(1-e^{-x})P(a)$, falls as $\gamma^2/(\gamma^2+\delta^2)$, so the trials rise with the detuning as $(\delta/\gamma)^2$, as for a weak sample.

\textit{Molecular bands and ring resonators.} A molecular band in solution is typically tens of terahertz wide and overlaps its neighbours, giving $(\delta/\gamma)^2$ of order unity, and its wings are not Lorentzian (Ref.~\cite{Whittaker2017} measured haemoglobin at $790$--$808\,$nm, away from its distinct features at $450$--$650\,$nm, and found a nearly flat spectrum). A ring resonator with the sample inside, the setting of Ref.~\cite{Belsley2022}, is not of this form either. Its transfer function, not its logarithm, is Lorentzian in the detuning, and the sample enters through the width, so the phase response to the intracavity loss vanishes on resonance and at half a free spectral range, where the transfer function is real. Detuning the ring moves information between the two terms of \cref{eq:disp-abs} while lowering their sum below its resonant value, and the ring's gain per photon entering the device is the resonant buildup, the interference between passes of \cref{sec:multipass}, not a detuning.

\textit{A caesium vapour cell.} On the caesium D2 line, of half-width $\gamma/2\pi=2.6\,\mathrm{MHz}$ ($3.0\,\mathrm{MHz}$ for rubidium), a room-temperature cell is Doppler-broadened: ten natural half-widths, $26\,\mathrm{MHz}$, is inside the Doppler core, where the phase term of \cref{eq:disp-abs} is $2\%$ of the absorption term, and the Lorentzian wing, where the ratio is $(\delta/\gamma)^2$, takes over only beyond about $0.7\,\mathrm{GHz}$; at $3\,\mathrm{GHz}$, detuned outside the ground-state hyperfine doublet, $(\delta/\gamma)^2\simeq10^6$. In the wing the absorption is the natural line's while the dispersion is that of the whole Doppler profile, so per absorbed photon the information is $\sigma_D\sqrt2/(\gamma\sqrt\pi)\simeq50$ times the resonant direct-transmission value of a weak sample, with $\sigma_D/2\pi=159\,\mathrm{MHz}$ the rms Doppler shift. With the absorptions registered, direct transmission of a weak sample gives $1/c^2$ at every detuning, and the gain in the wing is $1+[1-P(a)]\,x_{\mathrm L}$, where $x_{\mathrm L}$, the resonant optical depth of the same atoms without Doppler broadening, is about $\sigma_D\sqrt2/(\gamma\sqrt\pi)$ times the peak optical depth. That gain is about $44$ for the cell below at $1\,\mathrm{GHz}$, but close to one for a peak optical depth well below $0.02$. A cell of unit peak optical depth probed $1\,\mathrm{GHz}$ below the caesium D2 line, away from the other ground-state component, with $P(a)=0.1$ needs $2\times10^2$ absorbed photons for $1\%$ on $c$, hence $6\times10^6$ trials and under an hour at the heralded rate of \cref{sec:test}; at $3\,\mathrm{GHz}$ the absorbed photons are the same and the trials nine times as many.

\section{Many passes: the Zeno model and the ring resonator}
\label{app:zeno}

One cycle of the model of \cref{sec:multipass} maps the amplitudes $(r,s)$ in the reference and sample arms to $(r\cos\theta-s\sin\theta,\ \sqrt{t}[r\sin\theta+s\cos\theta])$ with $\theta=\pi/(2n)$, starting from $(1,0)$. With the two arms detected and the loss unregistered, the model's crossover, the transmittance above which it gives more Fisher information per absorbed photon than direct transmission, rises with the number of cycles, from $t=0.33$ at $10$ cycles to $t=0.85$ at $500$, and is lowest, $t=1/9$, for two cycles, where the information per absorbed photon is $2(1+2\sqrt t+2t)/[t(1-t)(3+2\sqrt t+t)]$ (one cycle is direct transmission). With only one arm detected, as in Fig.~9 of Ref.~\cite{Thomas2014}, the crossover at $10$ cycles is $t=0.56$ for either arm, close to their one half. Let $t\to1$ and $n\to\infty$ with $\nu=n(1-\sqrt{t})$ fixed. To first order in $1/n$ a cycle changes $r$ by $-\theta s$ and $s$ by $\theta r-(\nu/n)s$, so after $n$ cycles the amplitudes are the solution at unit time of
\begin{equation}
\dot r=-\tfrac\pi2\,s,\qquad\dot s=\tfrac\pi2\,r-\nu s,
\end{equation}
a two-level system driven through a quarter turn with its second level damped at the rate $\nu$. With $\Omega=\sqrt{\pi^2-\nu^2}$,
\begin{equation}
\begin{split}
r&=e^{-\nu/2}\Big[\cos\frac \Omega2+\frac\nu \Omega\sin\frac \Omega2\Big],\\
s&=e^{-\nu/2}\,\frac\pi \Omega\sin\frac \Omega2,
\end{split}
\end{equation}
continued to hyperbolic functions for $\nu>\pi$. The outputs have probabilities $r^2$ and $s^2$, the object absorbs $1-r^2-s^2$, and since $\partial_t=-(n/2\sqrt{t})\partial_\nu$ while the registered bound $1/[t(1-t)^2]$ tends to $n^2/(4\nu^2)$, the Fisher information per absorbed photon as a fraction of that bound is
\begin{equation}
G(\nu)=\frac{4\nu^2\,[(\partial_\nu r)^2+(\partial_\nu s)^2]}{1-r^2-s^2}.
\end{equation}
$G$ vanishes linearly as $\nu\to0$, where the photon is transferred to the sample arm and the scheme is a single transmission of a weak absorber, and tends to $(\pi^2/2+2)/\nu$ as $\nu\to\infty$, where the Zeno effect holds the photon in the reference arm. Between them it has one maximum, $G^*=0.931918$ at $\nu^*=2.012763$. The stationarity condition mixes $e^\nu$ with trigonometric functions of $\Omega$ and we have found no closed form for its root. At $t=0.95$ and $n=80$ the $n$-cycle model gives $0.93171$, and at $\nu=2$ it converges to $G(2)$ as $n$ grows. Away from the limit the discrete chain stays close to $G^*$: its best fractions are $0.914$, $0.925$, $0.930$ and $0.932$ at $t=0.3$, $0.5$, $0.7$ and $0.9$, reached at $n=3$, $6$, $11$ and $38$. For light passed $\ell$ times through the object itself, the detected fraction is $t^\ell$, the Fisher information per photon sent is $\ell^2t^{\ell-2}$, the object absorbs $1-t^\ell$, and the fraction of the registered bound is $\ell^2t^{\ell-1}(1-t)^2/(1-t^\ell)\to y^2/(e^y-1)$ with $y=\ell(1-t)$; its maximum is $y^*(2-y^*)=0.647610$ at $y^*=2+W_0(-2e^{-2})=1.593624$, with $W_0$ the principal branch of the Lambert function.

At $t=0.95$ and $n=80$ the coherent sum raises the first term of \cref{eq:manypass} to $47$ times the value $\sum_k|c_k|^2\|A_k\ket a\|^2/t$ the passes would give independently and $18.6$ times the direct-transmission reference $\sum_k|c_k|^2/t$; the two output ports alone then give $392$ units of Fisher information per absorbed photon, against $21$ for single-pass direct transmission and $421$ for the registered bound. The maximum $G^*$ lies on the underdamped side of critical damping, $\nu=\pi$. The best number of cycles is therefore $n^*\simeq2.01/(1-\sqrt t)$, which is $79$ at $t=0.95$ (the discrete model is best at $78$), and the $0.93$ of the two-port chain at $80$ cycles above is the optimum of the model. No readout of the two ports can do better than $G$, since $G$ already equals $4\|\partial_t\Psi_t\|^2$ over the absorption, the value \cref{eq:manypass} allows when the amplitudes are real. The remaining $7\%$ is not held by the absorbed outcome: at $\nu^*$ the chain's absorption is nearly stationary in $\nu$, so registering it would add only $0.003$ of the bound, bringing the three-outcome total to $0.935$. The function $y^2/(e^y-1)$ of repeated transmission is also the one that fixes the loss-limited optimum of Ref.~\cite{Wildfeuer2026b}. With the absorbed photons registered, the Zeno chain approaches the registered bound only asymptotically in the number of cycles [the large-$L$ expansion, Eqs.~(17) and (20), of Ref.~\cite{Facchi2002}, for binomial statistics]. Wildfeuer finds the Kwiat chain $1$--$9\%$ below it at $100$ cycles, the shortfall being the record of which pass absorbed the photon, and reports the bound as strict ``for arrangements in which the photon genuinely returns to the object'' \cite{Wildfeuer2026}. That holds for his examples, and at any finite number of passes whenever some pass's share of the absorption is not stationary in $t$, since that record then holds information; a ring resonator on resonance comes arbitrarily close to the bound for a long enough wavepacket (below).

For the ring resonator of \cref{sec:multipass}, off resonance at $t=0.9$ a detuning $\theta_{\mathrm{rt}}=0.01\,\mathrm{rad}$, a tenth of the half-width of the resonance, leaves $G=0.009$ at $r_{\mathrm c}=\sqrt{t}$, while overcoupling to $r_{\mathrm c}=\sqrt{t}-0.03$ keeps $G=0.89$. With a wavepacket of finite length, $G<1$, as for every arrangement with finitely many passes, since both have a first pass. The outputs, the absorbed outcome and the record of which pass absorbed the photon together cannot exceed the registered bound, the quantum Fisher information of the purified state \cite[Eq.~(24)]{Wildfeuer2026}, so $G=1$ needs the absorbed outcome and the record to hold no information. The record holds none only if every pass's share $|c_k|^2/\sum_j|c_j|^2$ of the absorption is stationary in $t$; since the first nonzero $c_k$ does not depend on $t$, $\sum_k|c_k|^2$ is then stationary, so the absorption $(1-t)\sum_k|c_k|^2$ is not, and the absorbed outcome holds information. The ring returns the photon to the object every round trip and loses nothing to that return: in its steady state the absorption is uniform in time, so the record of which pass absorbed the photon, which the object keeps \cite{Wildfeuer2026}, holds no information about $t$.

\textit{What lowers the ring's gain.} At $t=0.9$ the ideal ring on resonance gains $1/(1-t)=10$ over direct transmission per absorbed photon with the loss unregistered. Its information is held in the through port, which is dark at critical coupling, so light reaching that port from anywhere else lowers the gain and moves the best coupling away from critical. A background of $b$ counts per photon entering, compared with direct transmission with the same background, leaves a gain of $8.3$ at $b=0.01$ and $3.6$ at $b=1$. A Lorentzian source detected without time or frequency resolution, with the coupling re-optimised, keeps $90\%$ of the gain only if its linewidth is below $0.44\%$ of that of the critically coupled ring. Loss in the ring other than the object's, of round-trip transmittance $t_{\mathrm i}$, also lowers it. With the object spread round the ring, as in Ref.~\cite{Belsley2022}, and the coupling critical for the total loss, the gain per photon the object absorbs is $(1-t)t_{\mathrm i}\ln(tt_{\mathrm i})/[(1-tt_{\mathrm i})^2\ln t]$. For the ring of radius $50\,\mu\mathrm m$ in Ref.~\cite{Belsley2022}'s example, with their $2\,\mathrm{dB/cm}$ waveguide loss ($0.06\,\mathrm{dB}$ per round trip), this is $8.8$ at $t=0.9$. It falls fastest for the first fraction of a decibel ($9$ at $0.05\,\mathrm{dB}$ per round trip, $5$ at $0.46\,\mathrm{dB}$) and reaches direct transmission at $3.8\,\mathrm{dB}$ per round trip, or at $2.8\,\mathrm{dB}$ for a point object met before the waveguide loss. Eq.~(17) of Ref.~\cite{Belsley2022} gives the corresponding limit per photon entering, against the best single-pass probe. The Zeno model of \cref{sec:multipass} likewise assumes optics that lose photons only to the object.

\section{Chernoff information}
\label{app:chernoff}

With the photon heralded, the output probabilities with the object absent, $P(m)$, and present, $P(m|t)$ from \cref{eq:Pmt}, were computed with real amplitudes and $f(s)=\sum_{m:P(m)>0}P(m)^sP(m|t)^{1-s}$ was minimised numerically over $s\in[0,1]$; the Chernoff information is $C=-\ln\min_s f(s)$. Since $f$ is convex, the minimum sits at $s\to0^+$ whenever
\begin{equation}
\begin{split}
f'(0^+)=&\sum_{m:P(m)>0}P(m|t)\ln\frac{P(m)}{P(m|t)}\\
=&(1-t)P(a)-\tfrac12(1-t)^2\sum_{m:P(m)>0}\frac{\varrho(a,m)^2}{P(m)}\\
&+\tfrac14(1-t)^2P(a)\sum_{m:P(m)=0}|\braket{m}{a}|^2+O((1-t)^3)
\end{split}
\end{equation}
is non-negative, and then $C=-\ln f(0^+)$, which is \cref{eq:chernoff}. The threshold $t^*$ is the root of the first, unexpanded line, $f'(0^+)=0$; the truncated expansion has no root, being at least $(1-t^2)P(a)/2$ for every readout by the second sum rule.

\textit{A nearly dark port.} The expansion fails for an output whose occupation $P(m)$ is small beside its Elitzur--Vaidman probability $B_m=(1-\sqrt{t})^2|\braket{m}{a}|^2P(a)$. For real amplitudes on the side where its KD term is negative, $P(m|t)=[\sqrt{P(m)}+\sqrt{B_m}]^2\simeq B_m$, so its term in $f'(0^+)$ is $P(m|t)\ln[P(m)/P(m|t)]\simeq-B_m\ln[B_m/P(m)]$, which is the lowering quoted in \cref{sec:manyshot}; the bright outputs still contribute $(1-t)P(a)$ to first order in $1-\sqrt{t}$. The threshold is therefore where $(1-t)P(a)\simeq B_m\ln[B_m/P(m)]$. For the two-path interferometer of \cref{fig:fisher}(b) read out at $\delta\theta$ from the dark angle $\theta_0$, $P(m_-)=\sin^2\delta\theta$ and $B_-=(1-\sqrt{t})^2p\cos^2\theta$, so the condition is $1+\sqrt{t}=2(1-\sqrt{t})\cos^2\theta\,\ln[(1-\sqrt{t})\cos{\theta}\sqrt{p}/\delta\theta]$, and as $\delta\theta\to0$, where $\sqrt{t}\to1$ and $\cos^2\theta\to1-p$, this is the asymptotic form $1-\sqrt{t^*}\simeq1/\{[1-P(a)]\ln[(1-\sqrt{t^*})\sqrt{P(a)[1-P(a)]}/\delta\theta]\}$ of \cref{sec:manyshot}. The root of the asymptotic form is $t^*=0.55$ at $\delta\theta=0.001$ against the exact $0.56$, and $0.8012$ at $\delta\theta=10^{-6}$ against $0.8026$; the exact root of $f'(0^+)=0$ gives the values quoted in \cref{sec:manyshot}.

In the unregistered case the record is a set of independent Poisson counts with means proportional to $P(m)$ or $P(m|t)$, and the Chernoff information per photon sent is $C_{\mathrm P}=\max_s\sum_m[sP(m)+(1-s)P(m|t)-P(m)^sP(m|t)^{1-s}]$. For a weak absorber, with $P(m|t)-P(m)=-(1-t)\varrho(a,m)+O((1-t)^2)$ on the bright outputs, the quantity maximised is $s(1-s)A/2+(1-s)B+O((1-t)^3)$, where $A=(1-t)^2\sum_{m:P(m)>0}P(m)(\mathrm{Re}[\wv{m}])^2$ and $B=\sum_{m:P(m)=0}P(m|t)=(1-\sqrt{t})^2P(a)\sum_{m:P(m)=0}|\braket{m}{a}|^2$ is the null-port probability. Without null ports the maximum sits at $s=1/2$ and $C_{\mathrm P}=A/8$; per absorbed photon this is at most $\tfrac18(1-t)$ by \cref{eq:Fs1}, which agrees to this order with the direct-transmission value, $\tfrac18(1-t)+O((1-t)^2)$ per absorbed photon. With null ports the maximum moves to $s=\tfrac12-B/A$ and $C_{\mathrm P}=(A+2B)^2/(8A)$ for $B\le A/2$, and $C_{\mathrm P}=B$ otherwise. For the balanced Elitzur--Vaidman interferometer $A=(1-t)^2/4$ and $B=(1-\sqrt{t})^2/4\simeq(1-t)^2/16$, so $C_{\mathrm P}=\tfrac{9}{128}(1-t)^2$, which is $9/8$ of the direct-transmission exponent per absorbed photon: in the unregistered case the null port enters at the leading order, $(1-t)^2$, whereas in the heralded case it enters one order below the leading term $(1-t)P(a)$.

The largest gain of this kind comes from the readout of \cref{eq:qcb}, with $\ket\psi$ as one output. Its single bright output has $P(m)=1$ and $\wv{m}=P(a)$, so $A=(1-t)^2P(a)^2$, and its null ports have $B=S-\chi^2=(1-\sqrt{t})^2P(a)[1-P(a)]\simeq\tfrac14(1-t)^2P(a)[1-P(a)]$. The condition $B\le A/2$ is $P(a)\ge1/3$. Per absorbed photon and relative to direct transmission the exponent is therefore
\begin{equation}
\frac{8\,C_{\mathrm P}}{(1-t)^2P(a)}=
\begin{cases}
2\,[1-P(a)], & P(a)\le1/3,\\[2pt]
[1+P(a)]^2/[4P(a)], & P(a)\ge1/3,
\end{cases}
\end{equation}
to leading order in $1-t$: $9/8$ at $P(a)=1/2$ as above, $1$ for direct transmission, and a factor approaching two for weak coupling. Two is also the most that any measurement can give. The records are those of the coherent states $\ket{\sqrt N\psi}$ and $\ket{\sqrt NK_t\psi}$, whose quantum Chernoff exponent, both records being pure, is $-\ln|\braket{\sqrt N\psi}{\sqrt NK_t\psi}|^2=N(1-\sqrt{t})^2P(a)$ \cite{Audenaert2007}, or $(1-\sqrt{t})/(1+\sqrt{t})\simeq\tfrac14(1-t)$ per absorbed photon, twice the direct-transmission value. The factor of two between a null measurement and one limited by its Fisher information is the same one met for the heralded photon in \cref{sec:optimalreadout}.

\textit{Counterfactual ghost imaging.} Counterfactual ghost imaging \cite{Hance2021} fits the classification of \cref{sec:manyshot}. Each pixel of the object is either opaque or absent, so the task is discrimination against absence at $t=0$. The chained Zeno protocol makes the detector $D_1$ a null port that fires only when the pixel blocks, and cycling reduces the probability that a blocking pixel absorbs the photon towards zero, so the signal-to-noise ratio at fixed absorbed photon number grows without bound with the number of inner cycles in lossless optics (a loss per cycle instead sets a best number of cycles, as for any Zeno chain). That advantage is the null-port mechanism of \cref{fig:disc}(b) amplified by the Zeno effect. Ref.~\cite{Hance2021} noted that the same interference makes the protocol sensitive to the phase imprinted by a transparent object. Applied instead to the grey level of a translucent one, it would fall under \cref{cor:absorbed} for a single pass and, for the chain, under Wildfeuer's many-pass bound with the absorbed photons registered \cite{Wildfeuer2026}. The counterfactual gain of \cref{eq:Rweak} would improve the contrast of a single photon without improving the image built from many.

\section{How long the single-shot gain lasts}
\label{app:finiteM}
\label{sec:finiteM}

For $M$ heralded trials a non-disturbing definite-path scheme detects the object only by an absorption, with statistical distance $1-S^M$, and we take $R_M=\Delta_M/(1-S^M)$, so that $R_1=R$. The click records that favour ``present'' add to the distance, $\Delta_M=(1-S^M)+\sum_{\mathbf k}[P_1(\mathbf k)-P_0(\mathbf k)]_+$, where $\mathbf k$ is a record of $M$ clicks, $P_0$ and $P_1$ are its probabilities with the object absent and present, and $[x]_+=\max(x,0)$. For a weak absorber the likelihood ratio of a record is, by \cref{eq:LR}, $1-(1-t)Y$ with $Y=\sum_{i=1}^M\mathrm{Re}[\wv{m_i}]$ the sum of the real weak values of the outputs that clicked, and
\begin{equation}
\lim_{t\to1}R_M=1+\frac{\mathrm E_0\big[(-Y)_+\big]}{M\,P(a)},
\label{eq:RM}
\end{equation}
the expectation being over records with the object absent. The many-photon gain is the expected negative part of the summed weak values. For $M=1$ this is \cref{eq:Rweak}. For larger $M$ the sum rules fix the mean of $Y$ at $MP(a)$ and, for real KD terms without null ports, its variance at $MP(a)[1-P(a)]$, so negative sums become rare once the mean exceeds the spread. When $P(a)$ is small and no bright output is occupied more weakly than the object path, so that every port has clicked many times before this happens, $Y$ is Gaussian and the gain depends on the interferometer only through $x^2=MP(a)/[1-P(a)]$. It is below $10\%$ at $x^2=1$, below $1\%$ at $x^2=3$, and falls as $e^{-x^2/2}/x^3$ beyond (\cref{eq:RMgauss} below, \cref{fig:trials}): it is gone once three or four photons' worth of probability has been sent down the object path, of which the object has absorbed a fraction $1-t$. The same expectation fixes what negativity does to the error probability. The non-disturbing error is $\tfrac12S^M$, and negativity reduces it by the fraction $\sum_{\mathbf k}[P_1-P_0]_+/S^M\to(1-t)\,\mathrm E_0[(-Y)_+]$. For every readout $Y$ has mean $\mu=MP(a)$ and variance $\sigma^2\le MP(a)[1-P(a)]$, so
\begin{equation}
\begin{split}
\mathrm{E}_0[(-Y)_+]&=\tfrac12(\mathrm{E}_0|Y|-\mu)\le\tfrac12\big[\sqrt{\mu^2+\sigma^2}-\mu\big]\nonumber\\
&\le\frac{\sigma^2}{4\mu}\le\frac{1-P(a)}4,
\label{eq:quarter}
\end{split}
\end{equation}
which at $M=1$ is \cref{lem:negmax}. The bound is approached, after $M\simeq[1-P(a)]/(2\Nneg)$ trials, by an output with $P(m)\ll|\braket{m}{a}|^2P(a)$ and a large negative weak value, the imperfect null port of Appendix~\ref{app:sequential}, whose rare clicks make $Y$ Poisson rather than Gaussian; in the Gaussian regime the reduction is at most about a tenth of $1-t$ (below). For a weak absorber negativity therefore never lowers the error probability of a test by more than a quarter of $1-t$, and it does so when the test has barely begun to discriminate.

\textit{The Gaussian regime.} For real KD terms without null ports, when $P(a)$ is small and no bright output is occupied more weakly than the object path, the sum $Y$ of \cref{eq:RM} is Gaussian at large $M$, with mean $MP(a)$ and variance $MP(a)[1-P(a)]$, and
\begin{equation}
R_M\simeq1+\frac12\bigg[\sqrt{\frac2\pi}\,\frac{e^{-x^2/2}}{x}-\mathrm{erfc}\Big(\frac{x}{\sqrt2}\Big)\bigg],\quad x^2=\frac{MP(a)}{1-P(a)} .
\label{eq:RMgauss}
\end{equation}
\Cref{fig:trials} compares \cref{eq:RMgauss} with sums over all click records for the interferometers of \cref{fig:disc}. At $t=0.99$ the gains of all three fall with $x^2$ as \cref{eq:RMgauss} does, over six decades between them. The curve is the limit of small $P(a)$. The two-path points with $P(a)=1/20$ lie between $2\%$ and $84\%$ above it, the excess growing in the tail ($41\%$ at $x^2=8$, $84\%$ at $x^2=16$) and falling with $P(a)$ (from $41\%$ to $6\%$ at $x^2=8$ when $P(a)$ is reduced to $1/200$). The state $\ket{V_0}$, with $P(f)=1/3$, exceeds the curve by a factor of up to $56$ in the tail, where the gain is below $10^{-4}$. At $t=0$ the gain starts higher, from the values in \cref{fig:disc}(a), and ends sooner: both distances approach one, so $R_M-1\le S^M/(1-S^M)$, and what is left of the advantage is the ratio of the error probabilities, which grows as $e^{M(C+\ln S)}$. In the Gaussian regime the fraction $(1-t)\,\mathrm E_0[(-Y)_+]$ by which negativity reduces the error probability (above) is $(1-t)[1-P(a)]\,h(x)$, with $h(x)=x[\varphi_{\mathrm G}(x)-xQ(x)]$, $\varphi_{\mathrm G}$ and $Q$ being the standard normal density and upper tail probability. The maximum of $h$ is $0.101$, at $x=0.612$, the root of $\varphi_{\mathrm G}(x)=2xQ(x)$ (for the balanced two-path readout with $P(a)=1/20$ the sum over records gives $\mathrm E_0[(-Y)_+]/[1-P(a)]$ a maximum of $0.114$, at $M=7$). Outside this regime the bound \cref{eq:quarter} is approached by an imperfect null port: for the two-path readout with $P(a)=1/20$ at $0.003\,\mathrm{rad}$ from its dark angle, on the side where its KD term is negative, $\mathrm E_0[(-Y)_+]=0.235$ at $M=721$, against $[1-P(a)]/4=0.2375$.

\begin{figure}[t]
\centering
\includegraphics[width=\columnwidth]{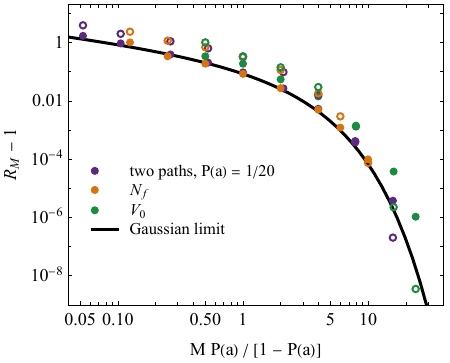}
\caption{Decay of the single-shot gain with the number $M$ of heralded trials: $R_M-1$ against $x^2=MP(a)/[1-P(a)]$ for the three interferometers of \cref{fig:disc} that have negative KD terms, from sums over all click records at $t=0.99$ (filled) and $t=0$ (open), and the limit \cref{eq:RMgauss} of a weak absorber and small $P(a)$ (black).}
\label{fig:trials}
\end{figure}

\section{Sequential tests, adaptive strategies and the evidence per absorbed photon}
\label{app:sequential}

The Chernoff exponent assumes a fixed number of trials and a fixed interferometer. A sensor can instead stop as soon as the evidence suffices, and can change the interferometer after every click. In the unregistered case the evidence per photon sent is the Poisson relative entropy $D=\sum_{m:\,P(m)>0}P(m)\,\Lambda(L_m)$, with $L_m$ the likelihood ratio \cref{eq:LR} of a click in output $m$ and $\Lambda(L)=L\ln L-L+1$; a null port has $L_m=\infty$, and its click is a proof. Nothing is absorbed when the object is absent, so what a sequential test spends is the expected number of photons absorbed when it is present. By Wald's identity, a test with false-alarm probability $\alpha$ that always ends by detecting a present object absorbs on average at least $\ln(1/\alpha)/\eta_{\max}$ photons, where $\eta_{\max}$ is the largest value over the available interferometers of the evidence per absorbed photon $\eta=D/[(1-t)P(a)]$. The one-sided sequential probability ratio test with the best fixed interferometer attains this, up to the usual overshoot, so choosing the interferometer adaptively gains nothing beyond the best fixed choice. Direct transmission gives $\eta_{\mathrm{dir}}=\Lambda(t)/(1-t)=\tfrac12(1-t)+\tfrac16(1-t)^2+\dots$, and for a geometry without null ports, expanding \cref{eq:LR} for a weak absorber gives
\begin{equation}
\begin{split}
\eta=&\frac{1-t}{2P(a)}\sum_{m:P(m)>0}P(m)(\mathrm{Re}[\wv{m}])^2\\
&+O\big((1-t)^2\big).
\label{eq:etaweak}
\end{split}
\end{equation}
The sum is the first form of \cref{eq:Fs1} without its null-port term, so $\eta\le\eta_{\mathrm{dir}}$ to this order for every such geometry, with equality when the KD terms are real: negativity redistributes the evidence as it redistributes the Fisher information. An ideal null port ends the test at its first click, after $(1+\sqrt{t})/[(1-\sqrt{t})|\braket{m}{a}|^2]$ absorbed photons on average; for the balanced Elitzur--Vaidman interferometer this is $2(1+\sqrt{t})/(1-\sqrt{t})$, against about $2\ln(1/\alpha)/(1-t)$ for direct transmission, so the null port wins once $\alpha$ is below about $e^{-(1+\sqrt{t})^2}$, and by an unbounded factor as $\alpha\to0$.

The sign of the weak values enters \cref{eq:etaweak} at second order: for real KD terms and no null ports, as shown below,
\begin{equation}
\begin{split}
\eta-\eta_{\mathrm{dir}}=&-\tfrac{(1-t)^2}{12P(a)}\sum_mP(m)\,\wv{m}\big(1-\wv{m}\big)^2\\
&+O\big((1-t)^3\big),
\label{eq:etasecond}
\end{split}
\end{equation}
so an output with a negative weak value raises the evidence, and an output with a weak value between $0$ and $1$, or above $1$, lowers it. By the sum rules a negative weak value is necessary for a gain at this order and sufficient when no weak value exceeds one, and the gain is of relative order $1-t$ (the values for the interferometers of \cref{fig:disc} are given below). A gain of order one at fixed $t$ needs an output with $(1-t)|\wv{m}|\gtrsim1$. Each term of $D$ is $P(m)\,\Lambda(|1-(1-\sqrt{t})\wv{m}|^2)$, and an output with $0\le\wv{m}\le1$ contributes at most its share $P(m)\wv{m}\Lambda(t)$ of the direct-transmission evidence whenever $t\ge e^{-2}$, as shown below. The outputs with $(1-\sqrt{t})|\wv{m}|\ll1$, whatever the sign or phase of their weak values, together contribute at most $\eta_{\mathrm{dir}}$ to $\eta$ to leading order in $1-t$, up to a relative correction of order $(1-\sqrt{t})|\wv{m}|$. The output's own amplitude $\braket{m}{\psi}$ must therefore be comparable to, or smaller than, the amplitude $(1-\sqrt{t})\braket{m}{a}\braket{a}{\psi}$ that the object adds. Such an output is the imperfect form of a null port. For the two-path interferometer of \cref{fig:fisher}(b) read out near its dark angle, with $u$ the ratio of the port's own amplitude to the one the object adds and the readout on the side where the port's KD term is negative,
\begin{equation}
\frac{\eta}{\eta_{\mathrm{dir}}}=1+\cos^2\theta_0\Big[(1+u)^2\ln\Big(1+\frac1u\Big)-u-\frac32\Big]+O(1-t),
\label{eq:etadarklead}
\end{equation}
which is $1.03$, $1.26$ and $2.24$ at $u=10$, $1$ and $0.1$ for $p=1/20$ (the full expression and the $t$ dependence are given below): the gain grows only as $\ln(1/u)$ as the port darkens, and on the other side of the dark angle the same amount to first order in $1/u$ is lost. A background rate $b_m$ in the port, from dark counts or imperfect interference, adds to $P(m)$ and to $P(m|t)$ alike and so sets a floor $u^2\ge b_m/[(1-\sqrt{t})^2|\braket{m}{a}|^2P(a)]$. This is the floor met for the Fisher information in \cref{eq:visibility}, and for a weak absorber it is demanding, since $(1-\sqrt{t})^2\simeq(1-t)^2/4$. Wildfeuer gives the corresponding ceilings on the exponent of the Zeno chain: $\ln[1/(2\sqrt\delta)]$ for a small crosstalk $\delta$ between the output ports, $\ln(1/\delta)$ for dark counts in the conclusive port, and $\ln[1/(1-\varepsilon_{\mathrm d})]$ for a detection efficiency $\varepsilon_{\mathrm d}$ \cite{Wildfeuer2026}.

Two of the statements above are proved here. First, for real KD terms without null ports the next term of \cref{eq:etaweak} is
\begin{equation}
\begin{split}
\eta-\eta_{\mathrm{dir}}=&\frac{(1-t)^2}{12}\Big[1-\frac1{P(a)}\sum_mP(m)\wv{m}^3\Big]\\
&+O\left((1-t)^3\right),
\label{eq:etacubic}
\end{split}
\end{equation}
and since the two sum rules give $\sum_mP(m)\wv{m}=\sum_mP(m)\wv{m}^2=P(a)$ in this case, the bracket equals $-\sum_mP(m)\wv{m}(1-\wv{m})^2/P(a)$, which is \cref{eq:etasecond}. A negative weak value is therefore necessary for a gain at this order, and sufficient when no weak value exceeds one, since then $\sum_mP(m)\wv{m}^3<\sum_mP(m)\wv{m}^2$. With a weak value above one the sign can go either way: $\wv{m}=(3,-1,1/2)$ with $P(m)=(1/50,1/18,208/225)$, which satisfies both sum rules and is realised by real amplitudes in a three-path interferometer, gives $\eta/\eta_{\mathrm{dir}}=0.994$ at $t=0.9$ and $0.93$ at $t=0.5$. For the two three-path states of \cref{fig:disc}, computed without expansion, the gain is $1.6\%$ ($\ket{N_f}$) and $1.8\%$ ($\ket{V_0}$) at $t=0.9$, and $11\%$ for both at $t=0.5$; for the two-path interferometer with $p=1/20$ it is $1.9\%$ and $12\%$. Second, an output with $0\le\wv{m}\le1$ contributes at most $P(m)\wv{m}\Lambda(t)$ to $D$ at all orders whenever $t\ge e^{-2}$: with $\epsilon=1-\sqrt{t}$, the function $\Phi(w)=\Lambda((1-\epsilon w)^2)$ has $\Phi''(w)=4\epsilon^2[1+\ln(1-\epsilon w)]\ge4\epsilon^2(1+\tfrac12\ln t)$ on $[0,1]$, so it is convex there and, vanishing at $w=0$, lies below the chord $w\,\Lambda(t)$. If every bright output has such a weak value, summing over $m$ and using $\sum_mP(m)\wv{m}=P(a)$ bounds the bright-output evidence by $\eta_{\mathrm{dir}}$ at all orders; without null ports the sum rules then leave only the values $0$ and $1$, for which $\eta=\eta_{\mathrm{dir}}$ at every $t$, so, for real KD terms, any gain over direct transmission comes from outputs with anomalous weak values.

For the two-path interferometer of \cref{fig:fisher}(b), with the readout on the side of the dark angle where the port brightens ($\wv{-}<0$), the dark port contributes
\begin{equation}
\begin{split}
\eta_-&=\frac{1-\sqrt{t}}{1+\sqrt{t}}\cos^2\theta\left[2(1+u)^2\ln\left(1+\frac{1}{u}\right)-1-2u\right],\\
u&=\frac{|\sin(\theta-\theta_0)|}{(1-\sqrt{t})\sqrt{p}\cos{\theta}},
\label{eq:etadark}
\end{split}
\end{equation}
where $(1-t)|\wv{-}|=(1+\sqrt{t})/u$. The bracket tends to $2$ as $u\to\infty$, which recovers the share $\cos^2\theta$ of \cref{eq:etaweak}, and grows as $2\ln(1/u)$ as the port darkens; with the bright port included, $\eta/\eta_{\mathrm{dir}}$ is $1.03$ at $u=10$, $1.26$ at $u=1$ and $2.24$ at $u=0.1$ for $p=1/20$ and $t=0.99$ (the limit $t\to1$ changes these by less than $1\%$). The leading term in $1-t$ is \cref{eq:etadarklead}, obtained from \cref{eq:etadark} with $\cos^2\theta\to\cos^2\theta_0$ and $\eta_{\mathrm{dir}}\to(1-t)/2$, plus the bright port's share $\sin^2\theta_0$; it gives $1.031$, $1.259$ and $2.236$. On the other side of the dark angle $\wv{-}=1/[(1-\sqrt{t})u]$ is positive, anomalous for $u<1/(1-\sqrt{t})$, and the port darkens for $u>1/2$; there $\eta$ falls below $\eta_{\mathrm{dir}}$ by the same amount to first order in $1/u$ and in $1-t$. The extinction ratio of the port is $P(m_-)/[(1-\sqrt{t})^2|\braket{m_-}{a}|^2P(a)]=u^2$.

In the registered case the first absorbed photon is itself a proof, and the question becomes how often one can decide before absorbing anything: a single pass offers only the null-port probability $(1-\sqrt{t})^2|\braket{m}{a}|^2P(a)$, whereas Mitchison and Massar showed that a Zeno protocol with many passes brings it arbitrarily close to one for any $t<1$ \cite{Mitchison2001}. Presence against absence of a weak absorber is also the case that Massar, Mitchison and Pironio placed outside their first absorption bound, with many coherent passes of a single photon as the route around it \cite{Massar2001}; their later bound covers transparencies of arbitrary separation but vanishes here, since one of the two absorption amplitudes is zero and their Eq.~(1) has the product of the two in its numerator \cite{Mitchison2002}.

\end{document}